\documentclass[a4paper,USenglish,thm-restate]{lipics-v2021}
\hideLIPIcs\nolinenumbers

\usepackage[utf8]{inputenc}
\usepackage{amsfonts}
\usepackage{hyperref}
\usepackage{xcolor}
\usepackage{algorithm}
\usepackage[noend]{algpseudocode}
\usepackage{amsthm}
\usepackage{ amssymb }
\usepackage{graphicx}
\usepackage{cite}

\title{Recoverable and Detectable \\ Self-Implementations of Swap}
\titlerunning{Recoverable and Detectable Self-Implementations of Swap}

\author{Tomer Lev Lehman}{Department of Computer Science, Ben Gurion University, Beer Sheva, Israel}{levletom@post.bgu.ac.il}{}{}
\author{Hagit Attiya}{Department of Computer Science, Technion, Haifa, Israel}{hagit@cs.technion.ac.il}{https://orcid.org/0000-0002-8017-6457}{}
\author{Danny Hendler}{Department of Computer Science, Ben Gurion University, Beer Sheva, Israel}{hendlerd@cs.bgu.ac.il}{https://orcid.org/0000-0001-7152-7828}{}
\authorrunning{T. Lev Lehman, H. Attiya, and D. Hendler}
\Copyright{Tomer Lev Lehman, Hagit Attiya, and Danny Hendler}

\begin{document}

\maketitle

\begin{abstract}
\emph{Recoverable} algorithms tolerate failures and recoveries of processes 
by using non-volatile memory.
Of particular interest are \emph{self-implementations} of key operations, 
in which a recoverable operation is implemented from its non-recoverable 
counterpart (in addition to reads and writes). 

This paper presents two self-implementations of the SWAP operation.
One works in the \emph{system-wide failures} model, 
where all processes fail and recover together, 
and the other in the \emph{independent failures} model, where each process 
crashes and recovers independently of the other processes.

Both algorithms are wait-free in crash-free executions, 
but their recovery code is blocking.
We prove that this is inherent for the independent failures model.
The impossibility result is proved for implementations of \emph{distinguishable}
operations using \emph{interfering} functions, and in particular, 
it applies to a recoverable self-implementation of swap.
\end{abstract}

\section{Introduction}

Recent years have seen a rising interest in the failure-recovery model for concurrent computing.
This model captures an unstable system, where processes may crash and recover. 
Two variants of the model have been considered. 
In the \emph{system-wide failure} model (also called the \emph{global-crash} model), all processes fail simultaneously and a single process is responsible for the recovery of the whole system.
In the \emph{independent failures} model (also called the \emph{individual-crash} model), each process can incur a crash independently of other processes and recovers independently.
\emph{Recoverable} algorithms, tolerating failures and recoveries, 
have been presented for various concurrent data structures,
for both the system-wide model~\cite{nahum2022recoverable, cho2022practical, friedman2018persistent, li2021detectable, FatourouKK2022performance, hoseinzadeh2021corundum, rusanovsky2021flat, correia2018romulus}
and the independent-failure model~\cite{nahum2022recoverable, cho2022practical, attiya2018nesting, li2021detectable,ben2019delay}.

The correctness of a recoverable algorithm can be specified in several ways. 
\emph{Durable Linearizability}~\cite{izraelevitz2016linearizability} 
intuitively requires linearizability~\cite{herlihy1990linearizability} 
of all operations that survive the crashes. 
\emph{Detectability}~\cite{friedman2018persistent} ensures that upon recovery, 
it is possible to infer whether the failed operation took effect or not and, 
in the former case, obtain its response.
\emph{Nesting-safe Recoverable Linearizability} (NRL)~\cite{attiya2018nesting}, defined for the independent failures model, ensures detectability and linearizability. 
It also allows the nesting of recoverable objects. 
By providing implementations of NRL primitive objects, a programmer can combine several of these primitives to create recoverable implementations of complex higher-level objects and algorithms. This level of abstraction can be helpful in the adoption of recoverable algorithms.

To facilitate high-level implementations of complex NRL objects 
it is helpful to introduce implementations of low-level base NRL objects. 
An attractive approach to designing low-level base NRL objects is through 
\emph{self-implementations}~\cite{nahum2022recoverable}, 
in which a recoverable operation is implemented with instances 
of the same primitive operation, 
possibly with additional reads and writes on shared variables.  
This approach ensures that when using the recoverable version of an operation, 
the system must only support its hardware-implemented primitive counterpart. 

NRL self-implementations already exist for various primitives,
including \emph{read}, \emph{write}, \emph{test\&set}, 
and \emph{compare\&swap}~\cite{attiya2018nesting, ben2019delay},
as well as \emph{fetch\&add}~\cite{nahum2022recoverable}. 
A universal construction using NRL read, write and \emph{compare\&swap} 
objects~\cite{ben2019delay} builds upon previously-introduced self-implementations 
of NRL objects to take any concurrent program with read write and CAS, 
and make it recoverable while adding only constant computational overhead. 

This paper presents the first NRL self-implementations of \emph{swap}, 
for both the system-wide and the independent failures models. 
Swap is a widely-used primitive that is employed by concurrent algorithms. 
Our implementations borrow ideas from the \emph{recoverable mutual exclusion (RME)}~\cite{golab2016recoverable} algorithms of \cite{golab2017recoverable, jayanti2019recoverable}, which use a similar approach to overcome swap failures.
Unlike these algorithms, however, 
our implementations are also challenged with the task of 
satisfying wait-freedom and linearizability.
Both our algorithms are wait-free in crash-free executions,
while the recovery code in both is blocking.

We also present an impossibility proof for implementing 
a class of \emph{distinguishable operations} 
using a set of \emph{interfering functions}~\cite{herlihy1991wait} 
in a recoverable \emph{lock-free} fashion in the independent failures model. 
In particular, this result applies to self-implementations of swap, but it also holds for,
e.g., implementing swap using fetch-and-add and swap combined. Other distinguishable operations to which this proof applies are the \emph{deque} of a queue object and the \emph{pop} of a stack object. 
Our impossibility result unifies and extends specialized 
results for self-implementations of \emph{test\&set}~\cite{attiya2018nesting} 
and \emph{fetch\&add}~\cite{nahum2022recoverable}. 
Another related impossibility result addresses 
recoverable consensus in the independent failures model~\cite{golab2020recoverable}.

Several previous papers introduce general mechanisms to port existing algorithms and 
make them persistent, e.g., 
by using transactional memory~\cite{chakrabarti2014atlas, correia2018romulus, izraelevitz2016failure, ramalhete2019onefile}, universal constructions~\cite{berryhill2016robust, cohen2018inherent, correia2020persistent}, 
or for specific families of algorithms \cite{ben2019delay,david2018log, friedman2020nvtraverse}.
Most of these transformations use strong primitives such as \emph{compare\&swap} 
while their non-recoverable counterparts may use weaker primitives, 
in terms of their consensus number~\cite{herlihy1991wait}.
We believe future research may use our self-implementation of swap to 
extend general constructions such as \cite{ben2019delay} 
mentioned above to programs that also use swap as a primitive.

Other papers present hand-crafted persistent implementations of specific data structures, 
e.g., \cite{friedman2018persistent, nawab2017dali, schwalb2015nvc, zuriel2019efficient}. 
In contrast to these implementations, our algorithms provide a recoverable counterpart to an atomic primitive operation, which we believe can later be used in various other implementations of recoverable algorithms. 
Similarly to NRL, \emph{detectable sequence specifications} (DSS), introduced by Li and Golab \cite{li2021detectable}, formalizes the notion of detectability. The DSS-based approach is more portable and less reliant on system assumptions in comparison to NRL, but delegates the responsibility for nesting to application code.

Our algorithm for the independent failures model uses an RME lock such as the one presented by Golab and Ramaraju~\cite{golab2016recoverable}, which uses only reads and writes. 
Additionally, a long line of papers on RME exists, solving several other aspects such as abortability, FCFS, and more~\cite{jayanti2017recoverable, KatzanM2021recoverable, jayanti2022recoverable, chan2020recoverable, jayanti2023constant, jayanti2019optimal}.

To summarize, our contributions are the following: 
\begin{itemize}
  \item A recoverable detectable self-implementation of swap in the system-wide failures model.
  \item A recoverable detectable self-implementation of swap in the independent failures model.
  \item An impossibility proof for implementations of distinguishable operations using interfering functions in the independent failures model.
\end{itemize}


\section{Related Work}
\subsection{Correctness conditions}
\label{sec:related_work:correctness_conditions}

Various correctness conditions and definitions exist for recoverable algorithms utilizing persistent memory. 
\emph{Strict linearizability}~\cite{aguilera2003strict} requires operations interrupted by a failure to take effect either before the failure or not at all. A relaxed condition called \emph{persistent atomicity}~\cite{guerraoui2004robust} in the context of message passing systems, allows an operation to take effect even after a failure, before the next operation invocation of the same process. \emph{Recoverable linearizability}~\cite{berryhill2016robust} builds on the definition of persistent atomicity and restores locality – the desirable property that an execution involving multiple objects is correct if and only if its projection onto each individual object is correct~\cite{golab2018recoverable}.

\emph{Durable linearizability}~\cite{izraelevitz2016failure}, defined for the system-wide failures model, assumes that, upon a crash, all processes fail and do not recover, and that new processes are spawned instead. Durable linearizability requires linearizablity of all operations surviving crashes, essentially requiring linearizability of the history when pending operations and their respective crash events are removed from the history.

\emph{Detectability}~\cite{friedman2018persistent} requires  that an object provides a mechanism that can tell, for every failed operation, whether or not it was completed, and if so obtain its response. Detectability can be added as an extra requirement for various correctness properties for implementations to satisfy

\emph{Detectable sequential specifications} (DSS)~\cite{li2021detectable} formalizes the notion of detectability by adding three new operations for every operation \emph{op}, \emph{prep-op}, \emph{exec-op}, and \emph{resolve}. Calling \emph{prep-op} notifies the system that it should remember the outcome of the upcoming operations, \emph{exec-op} executes the operation and \emph{resolve} can then be called after a crash to return the outcome of the most recently prepared operation (i.e., one for which \emph{prep-op} was called). DSS can be paired with any of the previously stated correctness conditions to obtain a detectable version of those properties. The DSS-based approach is more portable and less reliant on system assumptions in comparison to NRL, but delegates the responsibility for nesting to application code. In addition, the developer must implement 2 extra operations (\emph{prep-op, exec-op}) for each DSS-compliant operation.

\subsection{Recoverable Mutual Exclusion}
\label{sec:related_work:RME}

The \emph{Recoverable Mutual Exclusion(RME)} problem, formulated by \cite{golab2016recoverable}, extends the classic mutual exclusion problem for systems in which processes may crash and then recover. In addition to standard mutual exclusion properties such as deadlock-freedom and starvation-freedom, RME may also require the \emph{Critical Section Re-entry (CSR)} 
property, ensuring that if a process \emph{p} crashes within the critical section, then no other process can enter the critical section before  \emph{p} recovers and re-enters it~\cite{ben2022survey}. 

Several RME implementations were presented in recent years. Golab and Ramaraju presented an n-process RME implementation using only reads and writes~\cite{golab2016recoverable}. It is modeled after the non-recoverable ME algorithm of Yang and Anderson~\cite{yang1995fast}. The algorithm is based on a binary tournament tree of height $O(\log n)$ where each node is a two-process mutex (also introduced in~\cite{golab2016recoverable}).

Other RME implementations include a first-come-first-served (FCFS) RME algorithm with $O(\log n)$ Remote Memory Reference (RMR) complexity~\cite{jayanti2017recoverable}, 
an abortable sub-logarithmic RMR complexity RME algorithm~\cite{katzan2020recoverable}, 
and an abortable FCFS RME algorithm~\cite{jayanti2022recoverable}.

Golab and Hendler introduce an RME algorithm for the cache-coherent(CC) model~\cite{golab2017recoverable}, using $\text{fetch}\&\text{store}$ and $\text{compare}\&\text{swap}$, which incurs $O({\dfrac{\log n}{\log \log n}})$ RMRs. Their design is inspired by Mellor-Crummey and Scott’s queue lock (MCS lock)~\cite{mellor1991algorithms}. The MCS lock maintains a queue based structure that determines the order of entry into the CS. Each node in the queue also holds a boolean that is used to transfer ownership of the CS from a predecessor to a successor. Jayanti et al. \cite{jayanti2019recoverable} presented an RME algorithm with the same asymptotic RMR complexity for both the distributed shared memory (DSM) and the CC models.

\subsection{Recoverable implementations}
\label{sec:related_work:recoverable_implementations}

Various hand-crafted recoverable implementations exist for specific data structures. In \cite{friedman2018persistent}, the authors propose three different implementations of a concurrent lock-free queue, each satisfying different correctness properties, including durable linearizability and detectability. All three build on Michael and Scott’s queue and consist of a linked list of nodes that hold the enqueued values, as well as head and the tail references. In addition, there exist hand-crafted recoverable implementations of concurrent hash-maps~\cite{nawab2017dali, schwalb2015nvc, zuriel2019efficient}.

A different approach is to introduce general mechanisms to port existing algorithms and 
make them persistent. 
Mechanisms based on transactional memory~\cite{chakrabarti2014atlas, correia2018romulus, izraelevitz2016failure, ramalhete2019onefile} generally utilize various log-based methods to ensure durability of existing algorithms. 
Berryhill et al.~\cite{berryhill2016robust}
take Herlihy’s universal construction~\cite{herlihy1991wait} 
and transform it to satisfy recoverable linearizability. 
ONLL~\cite{cohen2018inherent} takes any deterministic object $O$ and produces a 
lock-free durably-linearizable implementation of $O$ that requires at most 
one persistent fence per update operation and no persistent fence for read-only operations. 

Another approach consists of mechanisms designed for specific families of algorithms. Ben-David et al.~\cite{ben2019delay} present a method that allows taking any concurrent program with reads, writes and CASs to shared memory and make it recoverable by utilizing code capsules. 
David et al.~\cite{david2018log} focus on link based data structures and provide generic techniques that enable designing what they call log-free concurrent data structures. 
Friedman et al.~\cite{friedman2020nvtraverse} present a general transformation that takes 
a lock-free data structure from a general class of node-based tree data structures 
(\emph{traversal data structures}) and automatically 
transforms it into a durably-linearizable implementation of the data structure.

\section{Model and Definitions}

We use a simplified version of the NRL system model~\cite{attiya2018nesting}. 
There are $n$ asynchronous \emph{processes} $p_1, \ldots, p_{n}$,
which communicate by applying atomic \emph{primitive} read, write and read-modify-write operations 
to \emph{base objects}. The state of each process consists of non-volatile \emph{shared variables}, 
which serve as base objects, as well as volatile \emph{local variables}. 

We first describe the \emph{independent failures} model, in which each process can incur a \emph{crash-failure} 
(or simply a \emph{crash}) at any point during the execution independently of other processes.
A crash resets all of its local variables to arbitrary values 
but preserves the values of all non-volatile variables. 

A process $p$ \emph{invokes an operation} $Op$ on an object by performing an \emph{invocation step}. 
\emph{Op completes} by executing a \emph{response step}, in which \emph{the response of OP is stored to a local volatile variable of $p$}. It follows that the response value is lost if $p$ incurs a crash before \emph{persisting} it, that is, before writing it to a non-volatile variable.
Operation $Op$ is \emph{pending} if it was invoked but was not yet completed;
a process has at most one pending operation. 

A \emph{recoverable operation Op} is associated with a \emph{recovery procedure} $Op.\text{RECOVER}$ that is responsible for completing $Op$ upon recovery from a crash. 
If the object only supports a single recoverable operation, then its recovery procedure is simply named RECOVER.
Following a crash of process $p$ that occurs when $p$ has a pending recoverable operation instance, the system  
eventually resurrects process $p$ by invoking the recovery procedure of the recoverable operation that was pending when $p$ failed. This is represented by a \emph{recovery step for $p$}. 

Formally, a \emph{history} $H$ is a sequence of \emph{steps}. There are four types of steps:
\begin{enumerate}
	\item An \emph{invocation step}, denoted $(INV, p, O, Op)$, represents the invocation by process $p$ of operation $Op$ on object $O$.
	\item A \emph{response step} $s$, denoted $(RES, p, O, Op, ret)$, represents the completion by process $p$ of operation $Op$ invoked on object $O$ by some (invocation) step $s'$ of $p$, with response $ret$ being written to a local variable of $p$;  \emph{s is the response step that matches s'}.
    An operation $Op$ can be completed either normally or when, following one or more crashes, the execution of $Op.\text{RECOVER}$ is completed. 
    \item A \emph{crash step} $s$, denoted $(CRASH, p)$, represents the crash of process $p$. We call the recoverable operation $Op$ of $p$ that was pending when the crash occurred the \emph{crashed operation of s}. $(CRASH, p)$ may also occur when the recovery procedure $Op.\text{RECOVER}$ is executed for recovering an operation of $p$ and we say that $Op$ is the crashed operation of $s$ also in this case.
	\item A \emph{recovery step $s$ for process $p$}, denoted $(REC, p)$, is the only step by $p$ that is allowed to follow a $(CRASH, p)$ step $s'$. It represents the resurrection of $p$ by the system, in which it invokes $Op.\text{RECOVER}$, where $Op$ is the crashed operation of $s'$. 
    We say that \emph{s is the recovery step that matches} $s'$.
\end{enumerate}

When a recovery procedure $Op.\text{RECOVER}$ is invoked to recover from a crash represented by step $s$, we assume it receives the same arguments as those with which $Op$ was invoked when that crash occurred.

As proven by \cite{ben2020upper}, detectable algorithms for the NRL model must keep an auxiliary state that is provided from outside the operation, either via operation arguments or via a non-volatile variable accessible by it. We assume that $Op.\text{RECOVER}$ has access to a designated per-process non-volatile variable ${SEQ}_p$, storing \emph{the sequence number of Op}.
Before $p$ invokes an operation on the object, it increments ${SEQ}_p$. 

An object is a \emph{recoverable object} if all its operations are recoverable. 
Below, we consider only histories that arise from operations on recoverable objects or atomic primitive operations.



Fix a history $H$. 
$H$ is \emph{crash-free} if it contains no crash steps (hence also no recovery steps). 
$H | p$ denotes the sub-history of $H$ consisting of all the steps by process $p$ in $H$. 
$H | O$ denotes the sub-history of $H$ consisting of all the invoke and response steps on object $O$ in $H$, 
as well as any crash step in $H$, by any process $p$, 
whose crashed operation is an operation on $O$ and the corresponding recovery step by $p$ (if it appears in $H$). 
$H|{<}p,O{>}$ denotes the sub-history consisting of all the steps on $O$ by $p$. 

A crash-free sub-history $H | O$ is \emph{well-formed}, 
if for all processes $p$, $H|{<}p,O{>}$ is a sequence of alternating, 
matching invocation and response steps, starting with an invocation step.
A crash-free history $H$ is \emph{well-formed} if: 
(1) $H | O$ is well-formed for all objects $O$, and 
(2) each invocation event in $H | p$, except possibly the last one, is immediately followed by its matching response step.

$H | O$ is a \emph{sequential object history} if it is 
an alternating series of invocations and the matching responses starting with an invocation;
it may end with a pending invocation. 
The \emph{sequential specification} of an object $O$ is the set of all \emph{legal} 
sequential histories over $O$. 
$H$ is a \emph{sequential history} if $H | O$ is a sequential object history for all objects $O$.

Two histories $H$ and $H'$ are \emph{equivalent}, 
if $H|{<}p,O{>}=H'|{<}p,O{>}$ for all processes $p$ and objects $O$. 
Given a history $H$, a \emph{completion} of $H$ is a history $H'$ 
constructed from $H$ by selecting separately, 
for each object $O$ that appears in $H$,
a subset of the operations pending on $O$ in $H$ and appending matching responses 
to all these operations, and then removing all remaining pending operations on $O$ (if any).

An operation $op_1$ \emph{happens before} an operation $op_2$ in $H$, 
denoted $op_1 <_H op_2$, 
if $op_1$'s response step precedes the invocation step of $op_2$ in $H$. 

\begin{definition} [Linearizability \cite{herlihy1990linearizability}, rephrased]
	\label{Definition: Linearizability}
A finite crash-free history $H$ is \emph{linearizable} if it has a completion
$H'$ and a legal sequential history $S$ such that 
  $H'$ is equivalent to $S$ and
$<_H \subseteq <_S$ (i.e., if $op_1 <_H op_2$ and both $op_1$ and $op_2$ appear in $S$, then $op_1 <_S op_2$).
\end{definition}

To define nesting-safe recoverable linearizability, 
we introduce a more general notion of well-formedness that applies also to 
histories that contain crash/recovery steps. 
For a history $H$, we let $N(H)$ denote the history obtained from $H$ by removing all crash and recovery steps.
%
A history $H$ is \emph{recoverable well-formed} if 
every crash step in $H | p$ is either $p$'s last step in $H$ or is followed in $H | p$ by a matching recovery step of $p$,
and 
$N(H)$ 
is well-formed.


\begin{definition} [Nesting-safe Recoverable Linearizability (NRL)]
\label{Definition:NRL}
A finite history $H$ satisfies \emph{nesting-safe recoverable linearizability} (NRL) if it is recoverable well-formed and $N(H)$ is a linearizable history.
An object implementation satisfies NRL if all of its finite histories satisfy NRL.
\end{definition}

We build upon the above definitions also for the \emph{system-wide failures} model, but we require that if a crash step occurs, then it occurs simultaneously for all processes whose operations crash. Formally, let $p_{i_1},\ldots, p_{i_k}$, for $k \leq n$, be the set of processes that have a pending invocation of operation $Op$ when a system-wide crash occurs. We represent the crash by appending the sequence $(CRASH, p_{i_1}), \ldots, (CRASH, p_{i_k})$ to the execution. During recovery, the system executes a parameterless \emph{global recovery procedure for Op} called $Op.\text{GRECOVER}$. We represent the recovery by appending the sequence $(REC, p_{i_1}), \ldots, (REC, p_{i_k})$ to the execution. Once $Op.\text{GRECOVER}$ completes, the system resurrects 
each of the processes $p_{i_1},\ldots, p_{i_k}$ for performing an \emph{individual recovery procedure for Op}, called $Op.\text{RECOVER}$. New operations on the object can be invoked only after recovery ends. If Op is the single object operation, we use the names GRECOVER and RECOVER instead of $Op.\text{GRECOVER}$ and $Op.\text{RECOVER}$, respectively.

An algorithm is \emph{lock-free} if, whenever a set of processes take a sufficient number of steps and none of them crashes, then it is guaranteed that one of them will complete its operation.
An algorithm is \emph{wait-free}, if any process that does not incur a crash during its execution completes it in a finite number of its steps. A swap object supports the \emph{SWAP(val)} operation, which atomically swaps the object's current value $cur$ to $val$ and returns $cur$.

\section{Detectable Swap Algorithm for the System-Wide Failures Model}
\label{sec:global-swap}

A key challenge to overcome when implementing a detectable swap object from read, write, and primitive swap operations is that the return values of one or more primitive swap operations may be lost upon a system-wide failure that occurs before the operations are persisted. These non-persisted operations may have already affected the state of the swap object and, moreover, operations by other processes may have already returned the values written by these primitive operations. To ensure linearizability, the implementation must identify such operations and handle them correctly.

The return value of each SWAP operation must meet a few requirements. 
First, it should be the input of another SWAP operation (or the initial value of the swap object). 
Second, the operand swapped in by one SWAP operation can be returned by at most a single other SWAP operation. 
Finally, in order to maintain linearizability, if $op_1 <_H op_2$ holds, 
then $op_1$ cannot return the value that was swapped in by $op_2$.

Figure~\ref{fig_swap_linearize} illustrates a scenario involving $6$ processes, denoted $p_1,\ldots p_6$, which perform $8$ SWAP operations, denoted $op_0, \ldots op_7$. A system-wide crash occurs when operations $op_0,op_2,op_4,op_6$ have already been completed (hence their return values are specified) while operations $op_1,op_4,op_5,op_7$ are pending. Note that operations $op_1,op_4,op_5$, although not completed, have surely affected the global state of the swap object as their inputs are the return values of other operations, while $op_7$ (pending as well) might or might not have affected the object's state. 

There are several ways the system may recover in order to produce a correct linearizable result. In all of them, $op_1$ must return $0$. The remaining operations might return different values in the following ways: (1) $op_4$ returns $2$, $op_5$ returns $3$, and $op_7$ returns $6$. (2) $op_7$ returns $2$, $op_4$ returning $7$, and $op_5$ returns $3$. (3) $op_4$ returns $2$, $op_7$ returns $3$, and $op_5$ returns $7$.
There are several possible linearizations in this example, because $op_7$ may be linearized in several ways since its effect on the global state is unknown. Note that for correctly recovering $op_1$, the operand of the very first operation, $op_0$, must be recorded. This is why our algorithms retain a record of all invoked operations.
\begin{figure}[bp]
\centering
\includegraphics[scale=0.525]{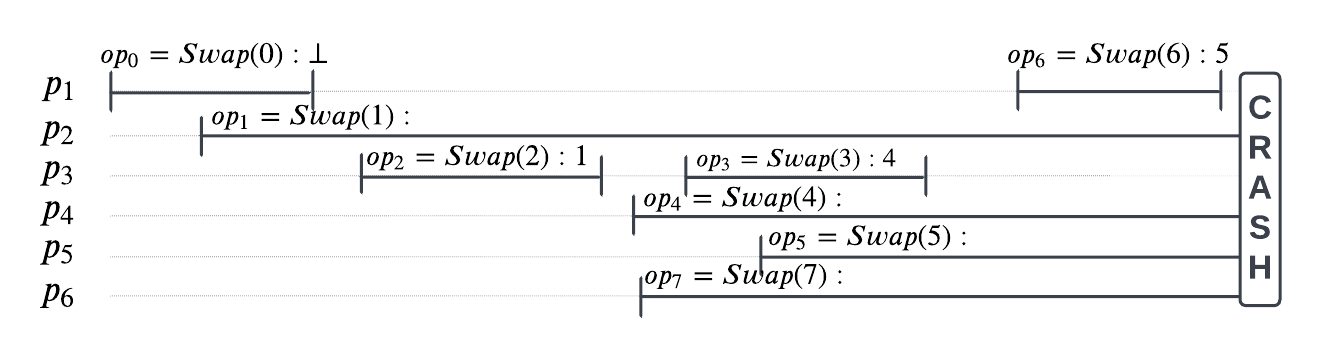}
\caption{An example of the effect of a system-wide failure.} 
\label{fig_swap_linearize}
\end{figure}

We represent the order of SWAP operations as a linked list of Node structures, 
the end of which is pointed by a $tail$ variable manipulated with primitive swaps. 
The list starts with a sentinel node called \emph{headNode}, 
which holds the object's initial value (denoted $\bot$).

Each Node structure represents a single SWAP operation and stores a pointer \emph{prev} to the node of its predecessor operation and the SWAP's operand \emph{val}. The order of SWAP operations is reflected by the order of the Node structures in the list. By doing so, each Node points to the previous Node structure that represents the previous SWAP operation, hence the operation's return value will be \emph{Node.prev.val}.

A problem can occur if a process successfully swaps its Node into the list but crashes before pointing from its structure to the previous Node. This type of failure may create what is referred to as \emph{fragments} in the list representing the SWAP operations. Thus, instead of a single complete list, crashes may result in several incomplete disconnected lists. In order to reconnect these fragments back to a complete list, our algorithms go over all previously-announced operations upon recovery and recreate a correctly ordered complete list of operations. 

A similar idea was used by the recoverable mutual exclusion (RME) algorithms of Golab and Hendler~\cite{golab2017recoverable} and Jayanti et al.~\cite{jayanti2019recoverable}, which also have to reconnect the fragments of an MCS lock \cite{mellor1991algorithms} linked-list based queue, caused by failures that occur just before or after primitive swap operations.

Our algorithms need to address two challenges that do not exist in the setting of \cite{golab2017recoverable,jayanti2019recoverable}, however. First, 
the SWAP operations of our algorithms are required to be wait-free whereas RME implementations are allowed to block. Second, unlike RME implementations, our algorithms are required to maintain linearizability. Specifically, the new order of list fragments, constructed during recovery, must respect the real-time order between SWAP operations. 

We address these challenges by employing a \emph{fragment ordering} scheme,
which we view as the key algorithmic novelty of our algorithms.
The scheme encapsulates the critical steps of each SWAP operation by two vector timestamp computations. 
Based on the resulting timestamps, the recovery code ensures the following invariant: 
if a fragment $A$ contains a Node $n_A$ that was created after an operation associated 
with a Node $n_B$ on fragment $B$ was completed, then fragment $B$ will be ordered after fragment $A$ in the connected list.
We formally describe this fragment ordering scheme in Definition~\ref{def_path_ordering}.

Figure~\ref{fig_swap_nodes} presents a set of fragments that may be generated immediately 
after the system-wide crash ending the execution depicted in Figure~\ref{fig_swap_linearize}. We describe it next and introduce a few terms that are used later in the sequel. The fragment of $op_7$ contains a single Node; we name such 1-size fragments \emph{singleNodes}. The node of $op_0$ belongs to the single fragment that contains \emph{headNode}; we name this fragment the \emph{head fragment}. The nodes of $op_5$ and $op_6$ belong to the single fragment that contains the node pointed to by \emph{tail}; we name this fragment the \emph{tail fragment}.
Fragments such as $(op_2, op_1)$ and $(op_3, op_4)$ that are neither singleNodes, nor head nor tail fragments are called \emph{middle fragments}. 

When ordering middle fragments and singleNodes, the algorithm uses vector timestamps for maintaining linearizability. As an example, consider a linked list, reconnecting the fragments of Figure~\ref{fig_swap_nodes}, in which $op_4.prev \leftarrow op_0$, $op_1.prev \leftarrow op_3$, $op_7.prev \leftarrow op_2$ and $op_5.prev \leftarrow op_7$. Although this list contains the Nodes of all operations from tail to head, it violates linearizability because $op_3$ is ordered after $op_2$ although it follows it in real-time order.
By using the two vector timestamps, our algorithm is able to order the fragments so that linearizability is maintained.

Another case that may arise is a set of fragments that consists of a single complete list with one or more singleNodes, which results from a crash that occurs when none of the pending operations completed their primitive swap operations. As we prove, in this case it is safe to put all these singleNodes (in any order) at the end of the list.

\begin{figure}[tp]
\centering
\includegraphics[scale=0.25]{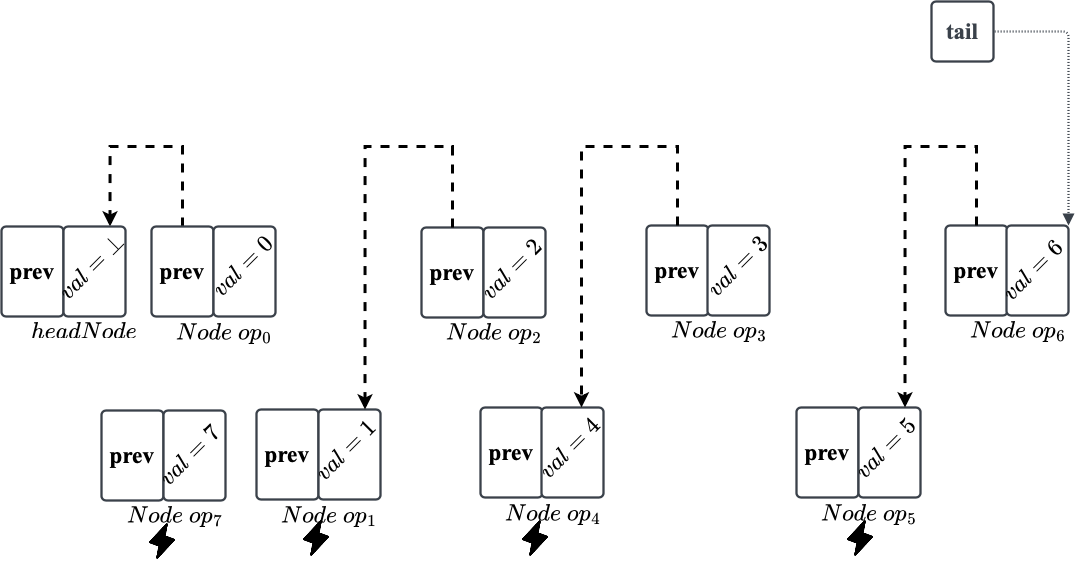}
\caption{List fragments existing after the execution described in Figure~\ref{fig_swap_linearize}, immediately after a system-wide crash. Operations 1,4,5 crashed after swapping their Nodes to \emph{tail} but before persisting their pointer to their predecessor Node.}
\label{fig_swap_nodes}
\end{figure}

\begin{algorithm}[ht]
\small
\caption{Recoverable detectable SWAP for system-wide failures model, code for process $i$.}\label{alg_glb_swap}
\begin{algorithmic}[1]
\Statex // Text in \textcolor{blue}{blue} is for the independent failures algorithm only.
\Statex \textbf{Define} Node: struct \{\emph{val} : \textbf{Value}, \emph{prev} : \textbf{ref} to Node, \emph{prevExecution} : \textbf{ref} to Node, \hspace*{3.45cm} \emph{startVts} : \textbf{vector}, \emph{endVts} : \textbf{vector}, \emph{seq} : \textbf{int}, \textcolor{blue}{\emph{inWork} : \textbf{int}}\}
\Statex \textbf{Initial State:}
\Statex $Nodes[i] \gets null$ for $i \in \{0,\ldots,n\}$
\Statex $VTS[i] \gets 0$ for $i \in \{1,\ldots,n\}$
\Statex $headNode \leftarrow new(Node)$
\Statex $headNode.val \leftarrow \bot$, $headNode.seq \leftarrow 0$, $headNode.prev \leftarrow null$, $headNode.prevExecution \leftarrow null$, \textcolor{blue}{$headNode.inWork \leftarrow 0$}.
\Statex $headNode.startVts \leftarrow collect(VTS)$, $headNode.endVts \leftarrow collect(VTS)$
\Statex $tail \leftarrow headNode$
\Statex $Nodes[0] \leftarrow headNode$
\Procedure {SWAP}{$val$} \Comment{executed by process $i$}
    \State $myNode \leftarrow new(Node)$ \label{SWAP:createNewNode}
    \State \emph{myNode.prev}$\leftarrow$null, \emph{myNode.seq}$\leftarrow SEQ_i$, \emph{myNode.prevExecution}$\leftarrow$ null, \emph{myNode.val}$\leftarrow$val \label{SWAP:InitNode}
    \State $VTS[i] \leftarrow VTS[i] + 1$ \label{SWAP:incVTSEntry}
    \State $myNode.startVts \leftarrow collect(VTS)$ \label{SWAP:collectTostartVTS}
    \State $prevExecution \leftarrow Nodes[i]$ \label{SWAP:readPrevExecution}
    \State $myNode.prevExecution \leftarrow prevExecution$ \label{SWAP:linkToPrevExecution}
    \State \textcolor{blue}{$myNode.inWork \leftarrow 1$} \Comment{begin swap}\label{ind_swap_inwork_1}
    \State $Nodes[i] \leftarrow myNode$ \Comment{announce the operation} \label{glb_declaration_point} \label{ind_swap_declared}
    \State $prev \leftarrow primitiveSwap(\&tail, myNode)$ \label{glb_lineraization_point} \label{ind_swap_primitive_swap}
    \State $myNode.prev \leftarrow prev$  \Comment{persisting operation}\label{glb_swap_persist_point} \label{ind_swap_prev_assign}
    \State $myNode.endVts \leftarrow collect(VTS)$ \label{ind_swap_collect_endVts}
    \State \textcolor{blue}{$myNode.inWork \leftarrow 0$} \Comment{finished operation}\label{ind_swap_inwork_0}
    \State \Return  $myNode.prev.val$ \label{glb_swap_return}
\EndProcedure
\algstore{globswap}
\end{algorithmic}
\end{algorithm}

\begin{algorithm}[ht!]
\small
\begin{algorithmic}[1]
\algrestore{globswap}
\Procedure {GRECOVER}{$ $} \label{RECVER:global-recover-start} \Comment{Global SWAP recovery procedure}
\State $V \leftarrow \emptyset$, $E \leftarrow \emptyset$ \label{RECOVER:startGraphConstruction}
\For{$i$ from $0$ to $n$} \label{glb_recover_for_loop_start}
    \State $currNode \leftarrow Nodes[i]$
    \While{$currNode \ne null$}
        \State $V \leftarrow V \cup \{currNode\}$
        \If {$currNode.prev \ne null$}
            \State $V \leftarrow V \cup \{currNode.prev\}$
            \State $E \leftarrow E \cup \{(currNode, currNode.prev)\}$
        \EndIf
        \State $currNode \leftarrow currNode.prevExecution$
    \EndWhile
\EndFor \label{RECOVER:endOfTraversalLoop}
\State $TAILNODE \leftarrow new(Node)$ \label{RECOVER:tailNodeCreate}
\State \Comment{Add graph node representing the list's tail and graph edge pointing to the \emph{tail} Node}
\State $V \leftarrow V \cup \{TAILNODE\}$ 
\State $E \leftarrow E \cup \{(TAILNODE, tail)\}$ \label{RECOVER:endGraphConstruction} \label{RECOVER:addTailNodeEdge}
\State Compute set $Paths$ of maximal paths in graph ${\cal{G}}=(V,E)$ \label{RECOVER:computePaths}
\label{glb_recover_compute_maximal_paths}
\State $MiddlePaths \leftarrow \emptyset $
\State $SingleNodes \leftarrow \emptyset$ \label{RECOVER:singleNodesInit}
\For {$path \in Paths$}
\If{$len(path) == 1$}
    \State $SingleNodes \leftarrow SingleNodes \cup \{start(path)\} $
    \State Remove $path$ from $Paths$ \label{RECOVER:singleNodesRemove}
\EndIf
\EndFor
\For{$path \in Paths$} \label{RECOVER:classifyFragmentsStart}
    \If{$TAILNODE \in path$ and $Nodes[0] \in path$}\label{glb_recover_if_single_full_path} 
    \Statex \Comment{There is a single full path from tail to head }
    \State For every Node in $SingleNodes$ re-execute SWAP from Line~\ref{glb_lineraization_point}\label{glb_recover_reexecute_SWAP_for_single_node}
    \State \Return \label{RECOVER:single_full_path_return}
    \ElsIf{$TAILNODE \in path$}
        \State $TailPath \leftarrow path$
    \ElsIf{$Nodes[0] \in path$}
        \State $HeadPath \leftarrow path$
    \Else
        \State $MiddlePaths \leftarrow MiddlePaths \cup \{path\}  $
    \EndIf
\EndFor \label{RECOVER:classifyFragmentsEnd}
\For{$curPath \in sort(MiddlePaths \cup SingleNodes$ in non-increasing $\succ$ order$)$}\label{glb_swap_path_sort} 
    \State $end(TailPath).prev \leftarrow start(curPath)$
    \label{glb_recover_prev_assign}
    \State update $TailPath$ to include added path \label{glb_swap_path_sort_end} 
\EndFor
\State $end(TailPath).prev \leftarrow start(HeadPath)$ \label{glb_swap_conclude_mending}
\label{glb_recover_tail_prev_assign}\label{glb_recover_final_line}
\EndProcedure \label{RECVER:global-recover-end}
\Procedure {RECOVER}{$val $} \Comment{Individual SWAP recovery procedure for process i} \label{glb_swap_indv_recover_start}
\State $myNode \leftarrow Nodes[i]$
    \If{$myNode == null$ or $myNode.seq < SEQ_i$}
        \State \Return SWAP($val$)
        \label{glb_recover_rerun_swap_fully}
    \Else
        \State \Return $myNode.prev.val$
        \label{glb_recover_return_ind_recover}
    \EndIf
\EndProcedure \label{glb_swap_indv_recover_end}
\end{algorithmic}
\end{algorithm}

\subsection{Detailed Description of the Algorithm}

Data structure definitions and the pseudo-code are presented by Algorithm~\ref{alg_glb_swap}. Text in \textcolor{blue}{blue} is for the independent failures algorithm and should be disregarded for now. We first describe key data structures and shared variables.

Each SWAP operation is represented by a single \emph{Node} structure. \emph{Node.val} stores the operand of the SWAP operation. \emph{Node.seq} stores the sequence number of the current process's SWAP operation. \emph{Node.prev} is a pointer to the \emph{Node} structure representing the previous SWAP operation. Consequently, \emph{Node.prev.val} stores the value that must be returned by the SWAP operation represented by \emph{Node}. Each \emph{Node} structure also stores two timestamp vectors of size $n$---\emph{Node.startVts} and \emph{Node.endVts}. Lastly, \emph{Node.prevExecution} is a pointer to the \emph{Node} structure representing the previous SWAP operation by the same process (if there is one).

\emph{Nodes} is an array of $n+1$ pointers to \emph{Node} structures.
\emph{Nodes}[0] points to the \emph{headNode} sentinel node. For each process $i \in \{1,\ldots,n\}$, \emph{Nodes}[$i$] points to the beginning of a list of \emph{Node} structures, induced by \emph{prevExecution} pointers, that represent the SWAP operations of process $i$. This array is used for recording all \emph{Node} structures created throughout the execution. 

\emph{VTS} is an array of length $n$ that serves as a global vector timestamp. Entry $i$ counts the number of operations performed by process $i$. \emph{tail} is a pointer to a \emph{Node} structure. The algorithm maintains a linked list of \emph{Node} structures representing the order of SWAP operations and tail points to the last \emph{Node} structure in the list.

The following order relation between paths is used by the global recovery procedure.

\begin{definition} 
\label{def_path_ordering}
Given two paths $A$ and $B$, we denote $A \succ B$ if there are nodes 
$n_A \in A$ and $n_B\in B$ such that $n_A.startVts > n_B.endVts$. If neither $A \succ B$ nor $B \succ A$ holds, we say that $A$ and $B$ are $\succ$-equal.
\end{definition}

A SWAP operation first creates a \emph{Node} structure and initializes it (Lines \ref{SWAP:createNewNode}-\ref{SWAP:InitNode}). It then increments its entry of the VTS, collects VTS, and writes the resulting vector timestamp to the \emph{startVTS} field of its node (Lines \ref{SWAP:incVTSEntry}-\ref{SWAP:collectTostartVTS}). In Lines \ref{SWAP:readPrevExecution}-\ref{SWAP:linkToPrevExecution}, the node representing the current operation is linked to the list of the previous operations executed by this process. Then, the process \emph{announces the operation} by writing a pointer to its node to its entry of the \emph{Nodes} array (Line~\ref{glb_declaration_point}). 
Next, the procedure invokes an atomic $primitiveSwap$ operation to read a node pointer from $tail$ and swap it with a pointer to the node representing the current operation (Line~\ref{glb_lineraization_point}). Then, the previous $tail$ value is persisted to the $prev$ field of the operation's \emph{Node} structure (Line \ref{glb_swap_persist_point}), thus adding this operation to the fragment of its predecessor operation. 
If a process executes Line~\ref{glb_lineraization_point} but the system crashes before it executes Line~\ref{glb_swap_persist_point}, a new fragment will result that cannot be reached from \emph{tail} using \emph{prev} pointers. (These fragments are reconnected by the global recovery procedure GRECOVER.) The operation terminates by performing a second collect of VTS, writing it to the \emph{endVts} field, and returning the previous value stored at $prev.val$ (Lines \ref{ind_swap_collect_endVts},\ref{glb_swap_return}).

The GRECOVER procedure (Lines \ref{RECVER:global-recover-start}-\ref{RECVER:global-recover-end}) of the SWAP operation is performed by the system upon recovery. As we've mentioned before, its task is to reconnect fragments caused by a system-wide crash by creating a total order between SWAP operations that maintains linearizability. When it terminates, all the SWAP operations that were previously announced (in Line \ref{glb_declaration_point}) are ordered in a single fragment that includes the \emph{headNode} and the node pointed to by \emph{tail}. As we prove, the order of operations induced by this fragment is a linearization of the execution. 

After the completion of GRECOVER, the system resurrects all the processes whose SWAP operations crashed, for executing the individual recovery procedure (Lines \ref{glb_swap_indv_recover_start}-\ref{glb_swap_indv_recover_end}). It first checks if the sequence number of the process' last announced operation (found in the \emph{Nodes} array) equals $SEQ_i$. If it does, the value of \emph{Nodes}[i]'s \emph{prev} field is returned (Line~\ref{glb_recover_return_ind_recover}). Otherwise, process $i$'s operation crashed before it was announced and so SWAP is re-executed (Line~\ref{glb_recover_rerun_swap_fully}).

Pending SWAP operations that updated their $prev$ pointers before the crash can simply return the value stored in $prev.val$. 
SWAP operations that did not update their $prev$ pointer in Line \ref{ind_swap_prev_assign} before a crash are of two types: Those that executed Line~\ref{glb_lineraization_point} before the crash and those that did not. SWAP operations of the latter type are simpler to deal with since they did not change the $tail$ pointer and can therefore be re-executed. Correctly ordering SWAP operations of the first type (i.e. those that executed Line~\ref{glb_lineraization_point} but did not execute Line~\ref{glb_swap_persist_point}) is more challenging, since their primitive swap operation changed \emph{tail}'s value but its return value was lost. 


We proceed to describe GRECOVER in more detail. It starts by constructing a directed graph $\cal{G}$ whose nodes correspond to Node structures and whose edges correspond to \emph{prev} pointers (Lines \ref{RECOVER:startGraphConstruction}-\ref{RECOVER:endGraphConstruction}). The construction is done by traversing (non-null) \emph{prev} and \emph{prevExecution} fields starting from each entry of the Nodes array (Lines \ref{RECOVER:startGraphConstruction}-\ref{RECOVER:endOfTraversalLoop}). After the traversal ends, a special TAILNODE node and an edge directed from it to the node pointed at by \emph{tail} are added to the graph for simplifying the handling of the tail fragment (Lines \ref{RECOVER:tailNodeCreate}-\ref{RECOVER:addTailNodeEdge}). A set \emph{Paths} of maximal directed paths in $\cal{G}$ is computed in Line \ref{RECOVER:computePaths}. $\cal{G}$ is cycle-free, because each SWAP operation performs Line \ref{glb_lineraization_point} at most once and, if it does, receives in response a pointer to an operation that performed Line \ref{glb_lineraization_point} before it. Thus, the set \emph{Paths} is well-defined. Each element of \emph{Paths} represents a fragment.

Next, all singleNodes (if any) are removed from \emph{Paths} and inserted into a separate \emph{SingleNodes} set (Lines \ref{RECOVER:singleNodesInit}-\ref{RECOVER:singleNodesRemove}). If \emph{Paths} has a fragment that contains both TAILNODE and the \emph{headNode} sentinel node then, as we prove, there are no middle fragments. In this case, each of the operations that correspond to \emph{SingleNodes} is executed, in turn, starting from Line \ref{glb_lineraization_point} and their nodes are thus appended to the end of this full path. Then GRECOVER returns (Lines \ref{glb_recover_if_single_full_path}-\ref{RECOVER:single_full_path_return}). Otherwise, the fragments in \emph{Paths} are categorized to a single \emph{HeadPath} fragment, a single \emph{TailPath} fragment, and a \emph{MiddlePaths} set that contains all other paths, which are middle fragments (Lines \ref{RECOVER:classifyFragmentsStart}-\ref{RECOVER:classifyFragmentsEnd}).

Next, all fragments other than \emph{HeadPath} and \emph{TailPath} are sorted in non-increasing $\succ$ order (see Definition~\ref{def_path_ordering}) and are appended, one after the other, to the end of \emph{TailPath} by updating appropriate \emph{prev} fields (Lines \ref{glb_swap_path_sort}-\ref{glb_swap_path_sort_end}). Sorting is done by comparing \emph{startVts} and \emph{endVts} fields as specified by Definition~\ref{def_path_ordering}. The construction of the full order is concluded by appending the \emph{HeadPath} to the end of the \emph{TailPath} (Line \ref{glb_swap_conclude_mending}).



The execution of GRECOVER, as well as that of the individual recovery procedure, can incur one or more crashes. In this case, the recovery process is re-executed upon recovery from each crash. As we prove, when it completes, operations are ordered correctly. 


\subsection{Proof of Correctness}

We say that the list $l=(a, a_1 = a.prev, a_2 = a_1.prev ..., a_m= a_{m-1}.prev, b=a_m.prev)$, for $m \geq 0$, of length $m+2$, is 
\emph{induced by prev pointers}. We also say that \emph{l starts at a} and \emph{l ends at $b$}.
We define the notion of a list \emph{induced by prevExecution pointers} similarly.


\begin{theorem}\label{thoerm_glb_swap_is_NRL} Algorithm~\ref{alg_glb_swap} implements a recoverable NRL SWAP in the system-wide failures model using only read, write and primitiveSwap operations. Its SWAP operations are wait-free.
\end{theorem}

\begin{proof}
Clearly from the code, SWAP operations are wait-free since they have no loops and so they terminate in crash-free executions.
We also observe that the algorithm is linearizable in crash-free executions: An operation $Op$ is linearized in Line~\ref{glb_lineraization_point} when it swaps a pointer to its Node to \emph{tail}. The next operation to execute Line~\ref{glb_lineraization_point} after $Op$ (if any) is guaranteed to return $Op$'s input as its response in line \ref{glb_swap_return}. If $Op$ is the first SWAP operation to perform Line~\ref{glb_lineraization_point} then, from the initialization of \emph{headNode}, it returns the object's initial value $\bot$. 


Next, we consider executions that incur crashes and prove the following property: Upon completion of the GRECOVER procedure, the list $\cal{L}$ induced by $prev$ pointers, starting from $tail$, contains exactly once the Node structure of every SWAP operation announced (in Line \ref{glb_declaration_point}) in the course of the execution and ends at \emph{headNode}. Moreover, the order induced by $\cal{L}$ respects operations' real-time order. 


We first now provide several lemmas and their respective proofs, Theorem~\ref{thoerm_glb_swap_is_NRL}'s proof continues following Lemma~\ref{paths_order_claim}.

The following lemma ensures that in GRECOVER at Line~\ref{glb_recover_compute_maximal_paths}, 
all Node structures ever announced by Line~\ref{glb_declaration_point} are in $V$, 
and all $prev$ pointers from Node $u$ to Node $v$ that were set in Line~\ref{glb_swap_persist_point} 
are represented by an edge $(u,v)\in E$.

\begin{lemma}
\label{lemma_all_nodes_and_pointer_are_in_graph}
    At Line~\ref{glb_recover_compute_maximal_paths} $\{v : v$ was announced by Line~\ref{glb_declaration_point}\} $\subseteq V$ and $\{(u,v) : u.prev == v\} \subseteq E$
\end{lemma}
\begin{proof}
    First Notice that every Node structure announced at Line~\ref{glb_declaration_point} is inserted into the Nodes array. Also, notice that in each cell in Nodes only a single process writes.

    In addition, every time a Node $x$ is overwritten by a new Node $y$ in Nodes[$p$] by process $p$, $y.prevExecution == x$. Therefore all of the Node structures created by $p$ can be reached by going over the list induced by $prevExecution$ pointers from $Nodes[p]$ for $p\in \{1\ldots n\}$. Also notice Nodes[0] is a special case of a list with constant length=1 representing the initial status.

    In the for loop at Line~\ref{glb_recover_for_loop_start} the recovery process goes over all process numbers from $1$ to $n$ and over $0$. For each process it goes over the list induced by its $prevExecution$ pointers, meaning $currNode$ gets assigned every Node ever announced by Line~\ref{glb_declaration_point} and Nodes[0], thus concluding that at Line~\ref{glb_recover_compute_maximal_paths} $\{v : v$ was announced by Line~\ref{glb_declaration_point}\} $\subseteq V$. In addition by going over every announced Node and adding the $TAILNODE$ link, we ensure that for every Node $u$, Node $v$ if $u.prev == v$ there is an edge $(u,v)\in E$ concluding that $\{(u,v) : u.prev == v\} \subseteq E$.
\end{proof}

The next lemma ensures every Node has at most one $prev$ or $tail$ pointer pointing at it.

\begin{lemma}\label{lemma_at_most_one_prev}
    For every Node $u$, either $tail = u$ and $\mid\{ v : v.prev = u\}\mid = 0$ or $\mid\{ v : v.prev = u\}\mid \leq 1$.
\end{lemma}
\begin{proof}
    We prove this lemma by reviewing all lines of code that assign prev pointers.
    First, notice that the lemma holds in the initial state since only $Nodes[0]$ exists and is pointed only by the tail pointer.

    When a prev pointer is assigned at Line~\ref{glb_swap_persist_point}, it is done using an atomic primitiveSwap operation, meaning that only a single SWAP operation can read the specific $prev$ value that was previously pointed by $tail$.

    When a $prev$ pointer is assigned during GRECOVER it is done in either Line~\ref{glb_recover_prev_assign}~or~\ref{glb_recover_tail_prev_assign}. In both cases, it is assigned with a Node that starts a path in the Graph$(V, E)$ and by lemma~\ref{lemma_all_nodes_and_pointer_are_in_graph} all Nodes created, and $prev$ pointers are represented in the graph meaning there is no other $prev$ pointer pointing to the Node assigned because it is a start of a maximal path in the graph $(V, E)$.
\end{proof}

We can now prove Lemma~\ref{glb_lemma_disjoint_paths}, showing that the set of paths $Paths$ computed in 
Line~\ref{glb_recover_compute_maximal_paths} are node-disjoint.

\begin{restatable}{lemma}{disjointPaths}
\label{glb_lemma_disjoint_paths}
Let $P$ and $J$ be two different paths that exist simultaneously in the global recovery process, and consider a Node $i \in P$, 
then $i \notin J$.
\end{restatable}

\begin{proof}
    From Lemmas~\ref{lemma_at_most_one_prev}~and~\ref{lemma_all_nodes_and_pointer_are_in_graph}, every Node $v\in V$ has at most one incoming edge, meaning there is at most one Node $u \in V$ such that $(u,v)\in E$. In addition, every edge $(u,v) \in E$ signifies that the node representing $u$ has its $prev$ pointer pointing to the node representing $v$. Thus, for every node $u \in V$ there is at most one $v \in V$ such that $(u,v) \in E$. Assume towards a contradiction that there exist two maximal paths $P$,$J \in Paths$, $P \ne J$ such that $i \in P$ and $i \in J$, then either there are nodes $u \in P$, $j \in J$ s.t $u \ne j$ and $(i,u),(i,j)\in E$, or there are nodes $u \in P$, $j \in J$ s.t $u \ne j$ and $(u,i),(j,i)\in E$. The first option provides a contradiction because node $i$ has two $prev$ pointers, and the second option means node $i$ has two prev pointers pointing at it in contradiction with Lemma~\ref{lemma_at_most_one_prev}.
\end{proof}

Lemma~\ref{after_recovery_lemma} shows that after a successful global recovery, the Node list starting at $tail$ is complete and holds all announced Node structures.

\begin{restatable}{lemma}{afterRecovery}
\label{after_recovery_lemma}
    At the end of a crash-free execution of the global GRECOVER procedure, there is a single list induced by $prev$ pointers starting from $tail$ and ending in $Nodes[0]$ such that all Node structures announced at Line~\ref{glb_declaration_point} are in it.
\end{restatable}

\begin{proof}
First, assume that the condition in Line~\ref{glb_recover_if_single_full_path} holds.
In this case, there is a maximal $path$ $P\in Paths$ that includes both TAILNODE and \emph{Nodes}[0]. Assume towards a contradiction that the lemma does not hold. It follows that there is a Node $i$ that isn't in $path$ $P$. There are two sub-cases to consider. Either there is another $path$ $J$ in $Paths$, or there isn't. If the latter sub-case holds, then $i$ must be in $SingleNodes$, hence SWAP will be re-executed starting from Line~\ref{glb_lineraization_point} on its behalf (in Line~\ref{glb_recover_reexecute_SWAP_for_single_node}), and because the execution is crash-free, \emph{i} will be inserted to the list induced by $prev$ pointers starting at $tail$, and this list will end in $Nodes[0]$ according to the code. The former sub-case is that $i \in J$ and $J$ is a \emph{middle fragment}, hence it is of length at least two. Let $x$ be the last Node in $J$ and let $Op$ be the operation represented by $x$. $Op$ must have executed Line~\ref{glb_lineraization_point} before the crash (otherwise no $prev$ pointer could point at $x$) but did not execute Line~\ref{glb_swap_persist_point} before the crash (since $x$ is the last node in $J$). Immediately after $Op$ executed Line~\ref{glb_lineraization_point}, \emph{tail} pointed to $x$. Since $Op$ did not execute Line~\ref{glb_swap_persist_point}, this contradicts the existence of a path in $Paths$ starting from $tail$ and ending in $Nodes[0]$.

Otherwise, the condition in Line~\ref{glb_recover_if_single_full_path} does not hold.
In this case, it follows from Lines \ref{RECOVER:classifyFragmentsStart}-\ref{glb_recover_final_line} that immediately after Line~\ref{glb_recover_tail_prev_assign} is executed by $Op$, TailPath is a list that starts from TAILNODE, ends with $Nodes[0]$, and contains all the Nodes in $V$. Thus by Lemma~\ref{lemma_all_nodes_and_pointer_are_in_graph}, this list contains all announced Nodes and the lemma holds.
\end{proof}

Lemma~\ref{paths_order_claim} ensures that the $\prec$ order can only hold in one direction for any two Paths.

\begin{restatable}{lemma}{glbPathOrder}
\label{paths_order_claim}
    Let $A$ and $B$ be two paths that exist simultaneously in the global recovery procedure. Then $A \prec B \implies B \nprec A$.
\end{restatable}

\begin{proof}
By Definition~\ref{def_path_ordering}, $A \prec B$ implies that there 
is a \emph{Node} $x \in A$ and a \emph{Node} $y \in B$ such that $x.endVts < y.startVts$. 
    From Lemma \ref{glb_lemma_disjoint_paths}, both fragments are disjoint. A new fragment is created only when a process executes Line~\ref{glb_lineraization_point} but does not execute Line~\ref{glb_swap_persist_point}. 
    This implies that all operations on fragment $A$ that completed Line~\ref{glb_lineraization_point} had done so before any of the operations on fragment $B$ performed Line~\ref{glb_lineraization_point}, or vice versa.
Since $x.endVts < y.startVts$ and $x.endVts$ is only written after executing Line~\ref{glb_lineraization_point}, and since $y.startVts$ is only set before executing Line~\ref{glb_lineraization_point} and after increasing $VTS$, 
it follows that all the operations of fragment $A$ executed Line~\ref{glb_lineraization_point}  
before any the operations of fragment $B$ have executed this line.

If we also have $B \prec A$, then by Definition~\ref{def_path_ordering}, 
there is a $Node$ $a \in A$, and a $Node$ $b \in B$, 
    such that $b.endVts < a.startVts$.
    It follows that the operation represented by $b$ executed Line~\ref{glb_lineraization_point}
    before the operation represented by $a$, which is a contradiction.
\end{proof}

To conclude the proof of Theorem~\ref{thoerm_glb_swap_is_NRL}, 
let $N_1, N_2 = N_1.prev, N_3 = N_2.prev  ... N_l = N_{l-1}.prev$ such that $N_1 = Tail$ and $N_l = Nodes[0]$ denote the Node structures in the list induced by \emph{prev} pointers beginning at \emph{tail}. In the initial state, $Tail = N_1 = N_l = Nodes[0]$. For a node \emph{x}, denote by \emph{proc(x)} the process that created and announced node \emph{x} and let $H = (SWAP_{proc(N_{l-1})}(N_{l-1}.val)$, $ SWAP_{proc(N_{l-2})}(N_{l-2}.val)$  $...$ $SWAP_{proc(N_{1})}(N_{1}.val))$, where $N_1, \ldots, N_l$ are the nodes in the single fragment that exists immediately after the GRECOVER procedure. Let $\alpha$ denote the execution that ends when GRECOVER completes. From Lemma~\ref{after_recovery_lemma}, $H$ is a sequential history that contains all the SWAP operations that were announced in $\alpha$. Moreover, for any extension $\beta$ of $\alpha$ in which all these SWAP operations return following the execution of their individual RECOVER procedures, they return the same values in $\beta$ and in $H$.
Operations that weren't announced in $\alpha$ but whose individual RECOVER procedures were executed in $\beta$, re-execute SWAP(val) (Line~\ref{glb_recover_rerun_swap_fully}) and are therefore linearized when they execute Line~\ref{glb_lineraization_point}.

It remains to show that for any two announced SWAP operations $\text{SWAP}_1, \text{SWAP}_2$, if $\text{SWAP}_1$ terminates in $\alpha$ before $\text{SWAP}_2$ starts, then $\text{SWAP}_1$ precedes $\text{SWAP}_2$ in $H$. Let $\text{Node}_1$, $\text{Node}_2$ be the Nodes created by $\text{SWAP}_1$ and $\text{SWAP}_2$, respectively.
Assume towards a  contradiction that $\text{SWAP}_2$ precedes $\text{SWAP}_1$ in $H$. Then there is a path induced by \emph{prev} pointers from $\text{Node}_1$ to $\text{Node}_2$ when GRECOVER terminates.
There are two cases to consider. The first case is that all the \emph{prev} pointers of the path between $\text{Node}_1$ and $\text{Node}_2$ were assigned by SWAP operations and not during recovery. This implies that between the time when $\text{SWAP}_1$ executed Line~\ref{glb_lineraization_point} and the time when $\text{SWAP}_2$ executed Line~\ref{glb_lineraization_point}, no process performed Line~\ref{glb_lineraization_point} without performing Line~\ref{glb_swap_persist_point}, otherwise $\text{Node}_1$ and $\text{Node}_2$ would have been on separate fragments just before recovery. Consequently, $\text{Node}_1$ and $\text{Node}_2$ are on the same fragment and $\text{Node}_1$ precedes $\text{Node}_2$ before a crash. It follows that $\text{Node}_1$ precedes $\text{Node}_2$ also in the single fragment that exists when GRECOVER terminates, hence $\text{SWAP}_2$ follows $\text{SWAP}_1$ in $H$. This is a contradiction.

The second case is that the path induced by prev pointers from $\text{Node}_1$ to $\text{Node}_2$ was formed during recovery. There are two sub-cases to consider. The first is that during recovery, one of the Nodes is in \emph{singleNodes} while the other is in a full path from TAILNODE to $\emph{Nodes}[0]$ (thus the condition in Line~\ref{glb_recover_if_single_full_path} is true). In this case, since operations that create Nodes in \emph{singleNodes} did not complete, it must be that $\text{Node}_2$ is in \emph{singleNodes} and $\text{Node}_1$ is on the full path. From Line~\ref{glb_recover_reexecute_SWAP_for_single_node}, $\text{SWAP}_2$ will be re-executed and so $\text{Node}_2$ will be placed before $\text{Node}_1$ in the list induced by \emph{prev} pointers starting from \emph{tail} at the end of GRECOVER, hence $\text{SWAP}_2$  follows $\text{SWAP}_1$ in $H$, a contradiction.

The second sub-case is when the condition of Line~\ref{glb_recover_if_single_full_path} is not satisfied.
This implies that just before crashing, $\text{Node}_1$ and $\text{Node}_2$ were on different fragments. Let $A$ and $B$ respectively denote the paths representing the fragments on which $\text{Node}_1$ and $\text{Node}_2$ were just before the crash. Since $\text{Node}_1.endVts < \text{Node}_2.startVts$ must hold, from Definition~\ref{def_path_ordering}, $A \prec B$ holds. Consequently, from Lemma~\ref{paths_order_claim} $B \nprec A$, therefore immediately after GRECOVER terminates there is a path from $\text{Node}_2$ to $\text{Node}_1$, hence $\text{SWAP}_2$ follows $\text{SWAP}_1$ in $H$, a contradiction.
\end{proof}

\section{Detectable Swap Algorithm for the Independent Failures Model} \label{section_ind_swap}

In the independent failures model, each process may crash and recover independently of other processes. A recoverable algorithm for this model must therefore allow one or more processes to execute RECOVER concurrently, while other processes may concurrently execute their SWAP operations. In order to handle this concurrency correctly, we introduce two key changes to Algorithm \ref{alg_ind_swap}. First, the RECOVER procedure now synchronizes concurrent invocations by using a starvation-free RME lock, implemented from reads and writes only, such as that proposed by \cite{golab2016recoverable}. This serializes the execution of the recovery code. 
The goal of the second change is to allow the recovery code to wait for a concurrent SWAP operation $Op$ to either complete or crash. Only once this happens, can the recovery code add the Node representing $Op$ to graph \emph{\cal{G}}. 

The pseudo-code of the RECOVER procedure for the independent failures model is presented by Algorithm \ref{alg_ind_swap}. The few additions done in the pseudo-code of SWAP are presented in \textcolor{blue}{blue} font in Algorithm~\ref{alg_glb_swap}. These consist of adding an \emph{inWork} field (initialized to $0$) to the \emph{Node} structure, setting it (in Line \ref{ind_swap_inwork_1}) just before the SWAP operation is announced in Line \ref{glb_declaration_point}, and resetting it (in Line \ref{ind_swap_inwork_0}) immediately after the \emph{endVTS} field is updated in Line \ref{ind_swap_collect_endVts}.


\begin{algorithm}[bt!]
\small
\caption{Recoverable detectable SWAP, for the independent failures model.}\label{alg_ind_swap}
\begin{algorithmic}[1]
\Procedure {RECOVER}{$val$} \Comment{executed by process i}
    \State $myNode \leftarrow Nodes[i]$
    \If{$myNode == null$ or $myNode.seq < SEQ_i$} \label{ind_recover_check_anounced_op}
        \State \Return SWAP($val$) \label{ind_recover_rerun_swap}
    \EndIf
    \State $myNode.inWork \leftarrow 2$ \label{ind_recover_announce_recovering}
    \State \emph{mutex}.\emph{lock}() 
    \label{ind_recover_acquire_lock}
    \If{$myNode.prev \ne null$} \label{ind_recover_if_prev_not_null}
        \State GoTo Line~\ref{ind_swap_recover_finished} \label{ind_recover_goto_return}
    \EndIf
    \State $V_1, E_1 \leftarrow gatherGraph()$ \label{ind_swap_recover_afterV1} 
    \State $tailNode \leftarrow tail$ 
    \label{ind_recover_read_tail_to_tailnode}
    \State $await(tailNode.inWork \in \{0,2\})$
    \State $V_2, E_2 \leftarrow gatherGraph()$
    \label{ind_recover_second_call_to_gather_graph}
    
    \State $V_2 \leftarrow V_2 \cup \{TAILNODE, tailNode\}$
    \State $E_2 \leftarrow E_2 {\cup} \{(TAILNODE, tailNode)\}$ 
    \label{ind_recover_add_tail_edge_to_e2}
    \State $V \leftarrow V_1 \cup V_2$\label{ind_swap_V_created}
    \State $E \leftarrow E_1 \cup E_2$
    \label{ind_recover_E_1_and_E_2_added_to_E}

    \State Compute set $Paths$ of maximal paths in graph ${\cal{G}}=(V,E)$
    \label{ind_recover_compute_maximal_paths}
    \For {$path \in Paths$}
        \If {$myNode \in path$}
            \State $myPath \leftarrow path$
        \EndIf
    \EndFor

    \If{$len(myPath) == 1$}
        \State re-execute SWAP from Line~\ref{ind_swap_primitive_swap} for $myNode$
        \State \emph{mutex}.\emph{release}()
        \State \Return $myNode.prev.val$
    \EndIf
    
    \State $MiddlePaths \leftarrow \emptyset $ \Comment{May include SingleNodes}
\For{$path \in Paths$}
    \If{$TAILNODE \in path$}
        \State $TailPath \leftarrow path$
    \ElsIf{$Nodes[0] \in path$}
        \State $HeadPath \leftarrow path$
    \Else
        \State $MiddlePaths \leftarrow MiddlePaths \cup \{path\}$
    \EndIf
\EndFor
\State $ordPaths \leftarrow sort(MiddlePaths)$ in non-increasing $\succ$ order \label{ind_swap_path_ordering} 
\label{ind_recover_j_index_orderedpathlist_declared}
\State $candidate \leftarrow $first path $C \in ordPaths$ after $myPath$ s.t. $start(C)\in V_1$ or null if no such $C$
\label{ind_recover_find_candidate}
\If{$candidate \ne null$}
    \State $myNode.prev \leftarrow start(candidate)$ \label{ind_recover_prev_assign}
\Else
    \State $myNode.prev \leftarrow start(HeadPath)$
    \label{ind_recover_prev_assign_headpath}
\EndIf
\State $myNode.endVts \leftarrow collect(VTS)$\label{ind_recover_collect_endVTS}\label{ind_swap_recover_finished}
\State \emph{mutex}.\emph{release}()
\State \Return $myNode.prev.val$
\label{ind_recover_return_at_end}
\EndProcedure
\end{algorithmic}
\end{algorithm}

\begin{algorithm}[th]
\small
\begin{algorithmic}[1]
\Procedure {gatherGraph}{$ $} \Comment{Used by Algorithm~\ref{alg_ind_swap}}
    \State $V \leftarrow \emptyset$
    \State $E \leftarrow \emptyset$
    \For{$j$ from $0$ to $n$}
        \State $currNode \leftarrow Nodes[j]$
        \While{$currNode \ne null$}
            \State $await(currNode.inWork \in \{0,2\})$
            \label{gather_graph_await_inwork}
            \State $V \leftarrow V \cup \{currNode\}$
            \If {$currNode.prev \ne null$}
                \State $V \leftarrow V \cup \{currNode.prev\}$
                \State $E \leftarrow E \cup \{(currNode, currNode.prev)\}$
            \EndIf
            \State $currNode \leftarrow currNode.prevExecution$
        \EndWhile
    \EndFor
\EndProcedure
\end{algorithmic}
\end{algorithm}

RECOVER first checks whether the \emph{Node} of the crashed operation was announced (Line~\ref{ind_recover_check_anounced_op}). If it wasn't, it re-executes SWAP(\emph{val}) (Line~\ref{ind_recover_rerun_swap}).
Otherwise, it signals that it is performing RECOVER by writing $2$ to the \emph{inWork} field of its Node and then attempts to acquire $mutex$ (Lines \ref{ind_recover_announce_recovering}-\ref{ind_recover_acquire_lock}). 
Next it checks if the operation already has a value to return and if so, returns this value (Lines~\ref{ind_recover_if_prev_not_null}-\ref{ind_recover_goto_return}, \ref{ind_swap_recover_finished}-\ref{ind_recover_return_at_end}). 

Lines~\ref{ind_swap_recover_afterV1}~-~\ref{ind_recover_compute_maximal_paths} construct the graph $\cal{G}$. 
Unlike the system-wide failure construction, we go over all Nodes \emph{twice}, 
thus constructing two sets of Nodes, $V_1$ and $V_2$. In addition, candidate paths chosen from $MiddlePaths$ are only chosen if the start of their fragment is from $V_1$ (Line~\ref{ind_recover_find_candidate}).
This is done because, after a single traversal that constructs $V_1$, there might be a Node in $V_1$ that is the start of a fragment that may be pointed by some $Node$ $x \notin V_1$. As we prove, a second traversal ensures that 
the problem cannot occur for a graph constructed based on $V = V_1 \cup V_2$. During each traversal of \emph{Nodes} the algorithm waits for each \emph{Node} $v$'s \emph{inWork} field, to be $0$ or $2$ before adding it to $V$ (Line~\ref{gather_graph_await_inwork} of \emph{gatherGraph}). This ensures that the operation $Op$ that created $v$ isn't concurrently executing its critical section of SWAP and therefore $v$ cannot change after being added to $V$.

The rest of the procedure is similar to that of GRECOVER in Algorithm~\ref{alg_glb_swap}. All maximal Paths in $\cal{G}$ are calculated and classified to \emph{TailPath}, \emph{HeadPath} and \emph{MiddlePaths}. In the end a single candidate path either from the sorted \emph{MiddlePaths} or the \emph{HeadPath} is selected to be linked to \emph{myNode} (Lines \ref{ind_recover_j_index_orderedpathlist_declared}-\ref{ind_recover_prev_assign_headpath}). Note that as we ensure only for Nodes in $V_1$ that any Node pointing at them is in $V$, only such Nodes are considered as candidates (Line~\ref{ind_recover_find_candidate}). As we prove, this ensures linearizability. 
The traversals that construct $V_1$ and $V_2$ are implemented by the helper function 
\emph{gatherGraph}.

\subsection{Correctness proof for the Swap implementation in the independent failures model} 
\label{appndx_swap_ind}

\begin{lemma}\label{lemma_ind_all_before_are_in_V1}
    Let $myOp$ be process $i$'s operation represented by $myNode$.
    In $i$'s execution of RECOVER, after completing Line~\ref{ind_swap_recover_afterV1} all Nodes representing operations that completed Line~\ref{ind_swap_primitive_swap} before $myOp$  completed it are in $V_1$.
\end{lemma}
\begin{proof}
    The lemma holds vacuously if \emph{myOp} did not execute Line~\ref{ind_swap_primitive_swap}.
    Let $Node_0$ be a Node created by $op_0$ performed by process $j$ that completed Line~\ref{ind_swap_primitive_swap} before $myOp$ completed it. Because $myOp$ completed Line~\ref{ind_swap_primitive_swap} before crashing, and $op_0$ completed Line~\ref{ind_swap_primitive_swap} before $myOp$, upon starting $i$'s RECOVER procedure $Node_0$ is in the List induced by $prevExecution$ pointers starting from $Nodes[j]$. When executing \emph{gatherGraph} in Line~\ref{ind_swap_recover_afterV1}, process $i$ goes over all Nodes in the list induced by \emph{prevExecution} pointers starting from \emph{Nodes}[$j$], therefore, $Node_0$ will be added to $V_1$.
\end{proof}

The following lemma ensures that every Node Structure has at most one prev pointer or tail pointing at it.
\begin{lemma}\label{lemma_ind_no_two_pointer_to_1_node}
     For every Node $u$, either $tail = u$ and $\mid\{ v : v.prev = u\}\mid = 0$ or $\mid\{ v : v.prev = u\}\mid \leq 1$.
\end{lemma}
\begin{proof}
    Similarly to Lemma~\ref{lemma_at_most_one_prev}, we prove this lemma by reviewing all lines of code that assign prev pointers.
    First, notice that the lemma holds in the initial state since only $Nodes[0]$ exists and it is pointed only by the tail pointer.

    As in the system-wide failures model algorithm, when a prev pointer is assigned at Line~\ref{ind_swap_prev_assign}, it is done using an atomic $primitiveSwap$ operation, meaning that only a single SWAP operation can read the specific $prev$ value that was previously pointed by $tail$.

    The difference from the system-wide model is that in the independent failures model algorithm a process might be in recovery while another process continues to execute SWAP. For this matter we must ensure that a recovering process does not assign a prev pointer to a Node that is about to be assigned by a process executing SWAP. 

    In RECOVER a $Node$ $v$'s prev pointer can be assigned in either Line~\ref{ind_recover_prev_assign}~or~\ref{ind_recover_prev_assign_headpath}.
    In both cases it is assigned by process $i$ to a Node that is a start of a maximal path in Graph ${\cal{G}}=(V,E)$, Let $y$ be the Node $v.prev$ is assigned to.
    It is left to show that for any Node $x$ either $y$ is not assigned to $x.prev$ at any stage or if  $x.prev==y$ then $(x,y)\in E$ during the execution of $i$'s RECOVER procedure meaning $y$ is not the start of a maximal path in ${\cal{G}}$.

    Let process $j$ be the process running the operation that created $x$.
    $j$ can assign a prev pointer to $x$ either in RECOVER or in SWAP.
    In the first case it will not be assigned to $y$ as it will await for process $i$ to release the $mutex$ lock before assigning a prev pointer and if the $mutex$ was released, then both $v$ and $y$ would be in $V_1$ for $j$'s recovery and therefore $y$ will not be the start of a maximal path in Graph ${\cal{G}}=(V,E)$ for $j$'s recovery.
    It is necessary to note that the Critical Section Re-entry (CSR) \cite{golab2016recoverable} property of the RME $mutex$ lock guarantees that if $i$ crashes during its recovery, $i$ is the only process allowed to acquire the $mutex$ lock upon its subsequent recovery.
    Specifically, $j$ will wait for $i$ to release $mutex$ even if $i$ crashes during its recovery.
    
   In the second case, during SWAP $y$ can be chosen to be assigned to $x.prev$ by $j$'s execution of Line~\ref{ind_swap_primitive_swap}, and is returned as the previous tail. 
   $y$ was also chosen to be assigned by $i$ during RECOVER, this can be done in either Line~\ref{ind_recover_prev_assign}~or~\ref{ind_recover_prev_assign_headpath}. In the latter case it is assigned to the start of the HeadPath while $v$ is the end of a failed fragment meaning the $TailPath$ can not be the $HeadPath$ because there is a failed fragment. Specifically $y$ cannot be pointed by $tail$ during $j$'s execution of Line~\ref{ind_swap_primitive_swap} as $y$ is the start of the $HeadPath$.
    
    The former case means $y$ is in $V_1$ for $i$'s recovery as only candidate paths from $MiddlePaths$ that start with a Node from $V_1$ are considered by Line~\ref{ind_recover_find_candidate}.
    Assume $y$ was also chosen by $j$ to be assigned to $x.prev$ when $j$ runs Line~\ref{ind_swap_primitive_swap}. This means that $x$ is not in $V_1$ for $i$'s recovery because if it was then $i$ would have waited for $x.inWork$ to be $0$ or $2$ before adding it to $V_1$ (Line~\ref{gather_graph_await_inwork} of gatherGraph) and when $x.prev$ is assigned to $y$, $x.inWork ==1$.
    Specifically that would mean that $(x,y)\in E_1$ for $i$'s recovery and $y$ is not the start of a maximal path.
    We conclude that during $i$'s first call to gatherGraph $x$ is not yet announced.
    
    $j$ chose $y$ to be assigned to $x$ during its primitveSwap of $tail$ meaning that when $y$ is added to $V_1$, $tail == y$. It follows that when $i$ reads $tail$ to $tailNode$ (Line~\ref{ind_recover_read_tail_to_tailnode}), either $tail=y$ (meaning $tailNode$ also equals $y$) or $x$ already performed its primitiveSwap.
    Here if $tailNode == y$ then an edge $(TAILNODE, y)$ is added to $E_2$ (Line~\ref{ind_recover_add_tail_edge_to_e2}).
    Otherwise during the second call to gatherGraph (Line~\ref{ind_recover_second_call_to_gather_graph}) $x$ already performed its primitiveSwap and $i$ will wait for $x.inWork$ to be $0$ or $2$ (Line~\ref{gather_graph_await_inwork} of gatherGraph) therefore it will wait for $x.prev$ to equal $y$ and $(x,y)$ will be in $E_2$. In both cases $E_2$ is then added to $E$ (Line~\ref{ind_recover_E_1_and_E_2_added_to_E}), concluding that there is an edge in $E$ pointing to $y$. It follows that $y$ is not a start of a maximal path in ${\cal{G}}$ for $i$'s recovery.
    
\end{proof}

\begin{lemma}\label{lemma_ind_ordering_is_correct}
    For any $Node$ $x$ and any $Node$ $y\ne x$ such that $y$ is in the list induced by $prev$ pointers starting from $x$, $y.startVts \ngtr x.endVts$.
\end{lemma}
\begin{proof}
    Let $l$ be the list induced by $prev$ pointers from $x$ to $y$.
    We split the proof to 2 cases. The first is in $l$ all $prev$ pointers were assigned in SWAP by Line~\ref{ind_swap_prev_assign}. The second is that there exist $Nodes$ $i,m$ such that $i,m \in l$ and $i.prev == m$ and $i.prev$ was not assigned by Line~\ref{ind_swap_prev_assign}.

    For the first case the proof is straightforward as that means $y$'s operation has completed Line~\ref{ind_swap_primitive_swap} before $x$'s operation completed it. Therefore $y.startVts$ was collected (Line~\ref{SWAP:collectTostartVTS}) before $x.endVts$ was collected by Line~\ref{ind_swap_collect_endVts}. Since VTS is incremented before collecting $y.startVts$ (Line~\ref{SWAP:incVTSEntry}) and VTS can only be incremented, it follows that $y.startVts \ngtr x.endVts$.

    For the second case note that $i$ is in the list induced by $prev$ pointers starting at $x$ and $y$ is in the list induced by $prev$ pointers starting at $m$. When $i.prev$ is assigned to $m$ during recovery it is done in either Line~\ref{ind_recover_prev_assign}~or~\ref{ind_recover_prev_assign_headpath}. If it is assigned in Line~\ref{ind_recover_prev_assign} it is either done to a candidate path that is smaller or equal by $\succ$ order (Definition~\ref{def_path_ordering}) meaning $y.startVts \ngtr x.endVts$, or $i$ is on the $TailPath$. If $i$ is on the $TailPath$ then $x$'s operation performed Line~\ref{ind_swap_primitive_swap} after any $Node$ $y$ that is in the list induced by $prev$ pointers starting from $m$ performed it also concluding $y.startVts \ngtr x.endVts$.
    
    Otherwise it is assigned to the $HeadPath$. Operations on the $HeadPath$ either succeeded and completed Line~\ref{ind_swap_primitive_swap} before $x$ ended meaning $y.startVts \ngtr x.endVts$, or failed and were mended to fragments that were eventually mended to the $HeadPath$. In the latter case when those fragments were mended either $x$ was announced or not, if it was then their fragment was smaller or equal according to $\succ$ order (Definition~\ref{def_path_ordering}) than $x$'s fragment (otherwise they would have eventually be mended to $x$'s fragments and not the $HeadPath$), meaning for any Node $y$ on their fragment $y.startVts \ngtr x.endVts$. 
    If $x$ was not announced when those fragments were mended, then also for any Node $y$ on their fragment $y.startVts \ngtr x.endVts$ because $y$ started before $x$ was announced and before $x.endVts$ was assigned and VTS can only be incremented.
%
%
\end{proof}

\begin{theorem}
    Algorithm~\ref{alg_ind_swap} implements a recoverable NRL SWAP in the independent failures model using only read, write and primitiveSwap operations and satisfies NRL. Its SWAP operations are wait-free.
\end{theorem}

\begin{proof}
    Proving correctness for crash-free executions, and wait-freedom of SWAP operation can be done exactly as in the proofs for the system-wide failure model. We now consider an execution $\alpha$ with independent process crashes.

    Let $node(Op)$ denote the node that represents $Op$. We show that the (possibly partial) order that exists between operations which is induced by \emph{prev} pointers, always satisfies the following requirements. 1) Each operation $Op$ can only return the value of the operation $Op'$ represented by $node(Op).prev$ and, 2) if it does, it is the only operation that returns the value of $Op'$ and it does not precede in real-time order any $Op''$ in the list induced by $prev$ pointers from $node(Op)$. 

    For any operation $Op_1$ that returns the value of operation $Op_0$, From Lemma \ref{lemma_ind_no_two_pointer_to_1_node}, no other SWAP operation can return $node(Op_0).val$.
    
    


    It is left to show that in any execution $\alpha$ if $op_1$ ended before $op_2$ started, then $node(op_2)$ cannot be in the list induced by $prev$ pointers starting from  $node(op_1)$ at any point during the execution. Proving this will establish that the operations can be linearized correctly according to the reversed order of the list induced by prev pointers, as in the system-wide failures model algorithm.
    
    As $op_1$ ended before $op_2$ started, and any operation writes to its node's $endVts$ before completing,  $node(op_2).startVts > node(op_1).endVts$ holds. It now follows from Lemma~\ref{lemma_ind_ordering_is_correct} that if  $node(Op_2)$ is in the list induced by $prev$ pointers starting from $node(Op_1)$, then $node(Op_2).startVts \ngtr node(Op_1).endVts$. This is a contradiction.
\end{proof}

\section{Impossibility of lock-freedom for the independent failures model}
 In this section, we prove a theorem establishing the impossibility of implementing lock-free algorithms for a wide variety of recoverable objects under the independent failures model. 
 This generalizes previous results~\cite{attiya2018nesting, nahum2022recoverable} to a wider family of operations and implementations. 
 In particular, it applies to any self-implementation of swap under the independent failures model, showing that our usage of a mutual exclusion lock in Algorithm~\ref{alg_ind_swap} is essential.

 We start by defining the notion of a \emph{distinguishable operation}. 
 An operation $M$ is distinguishable, if there exists a history $\alpha_{base}$ in $spec$ and two invocations $M_p$ and $M_q$ of $M$, such that the return values of the invocations allows the system to distinguish which operation is applied right after $\alpha_{base}$. Formally:
\begin{definition}[Distinguishable operation]
\label{def_distinguishable}
 Operation $M : \text{VAL} \rightarrow RET$ is \emph{distinguishable} if there exists a history $\alpha_{base}$ and values $x,y \in \text{VAL}$, $z\in RET$, such that if $M(x)$ and $M(y)$ are applied sequentially right after $\alpha_{base}$, the first (and only the first) invocation of $M$ to complete returns $z$.
\end{definition}

Assume a swap object with value 0 and two SWAP operations.
If SWAP(1) and SWAP(2) are applied sequentially, 
only the first operation applied will return 0. 
This shows that SWAP is a distinguishable operation. Similarly, it's easy to show that pop and deque operations of the stack and queue data structures, as well as \emph{fetch\&add} and \emph{test\&set}, are also distinguishable operations.

Our impossibility result applies to implementations of distinguishable operations 
that use only read, write, and a set of \emph{interfering functions},
defined as follows:

\begin{definition}[Interfering functions~\cite{herlihy1991wait}]
\label{def:interfering}
Let F be a set of primitive functions indexed by an arbitrary set $K$. Define F to be a set of \emph{interfering functions} if for all $i$ and $j$ in $K$, for any object $O$ that supports $f_i$ and $f_j$, and for any state $S$ of $O$, 
\begin{description}
\item[(1) $f_j$ and $f_i$ commute:]
 The application of $f_i$ to $O$ in state $S$ by process $p$ followed by the application of $f_j$ to $O$ by process $q$ leaves $O$ (but not necessarily the local state of each process) in the same state as the application of $f_j$ to $O$ in state $S$ by process $q$ followed by the application of $f_i$ to $O$ by process $p$; or
\item[(2) $f_j$ overwrites $f_i$:]
The application of $f_i$ to $O$ in state $S$ by process $p$ followed by the application of $f_j$ to $O$ by process $q$ leaves $O$ (but not necessarily the local state of each process) in the same state as the application of $f_j$ to $O$ in state $S$ by $q$ alone.
\end{description} 
\end{definition}

A \emph{configuration} $C$ consists of the states of all processes 
and the values of all shared base objects. 
Sometimes we use the notions of a configuration and a history interchangeably. 
For example, if a finite history $H$ leads to a configuration $C$ 
we may use $H$ for representing $C$ when $H$ is clear from the context.  
Two configurations $C_1$ and $C_2$ are \emph{indistinguishable} 
to a set of processes $P$, denoted $C_1 \stackrel{P}{\sim} C_2$, 
if every process in $P$ has the same state in $C_1$ and $C_2$, 
and every shared object holds the same value in $C_1$ and $C_2$.

Given a configuration $C$ reached after a history $\alpha_{base}$,
distinguishable operation $M$, and a process $r\in \{p,q\}$, 
we say that $C$ is \emph{$r$-valent} if there is an execution starting from $C$ 
in which the return value of $M$ or $M$.RECOVER by $r$ is $z$ (where $\alpha_{base}$, 
$M$ and $z$ are as in Definition~\ref{def_distinguishable}). 
$C$ is \emph{bivalent} if it is both $p$-valent and $q$-valent,
for $p \neq q$. 
$C$ is \emph{$p$-univalent} if it is $p$-valent and not $q$-valent, 
and symmetrically for $q$-univalent. 
$C$ is \emph{univalent} if it is either $p$-univalent or $q$-univalent. 
Let $C$ be a bivalent configuration and $s$ be a step.
If $C \circ s$ is univalent, we say that $s$ is a \emph{critical step}. We generalize the proofs of~\cite{attiya2018nesting, nahum2022recoverable} 
to prove the following theorem by using valency arguments~\cite{fischer1985impossibility, golab2020recoverable}.

\begin{theorem} 
\label{impossibility_theorm}
Let M be a distinguishable operation. 
There is no recoverable implementation $I$ of $M$ from read, 
write and a set of $K \geq 1$ of interfering primitive operations 
$f_1 \ldots f_K$ in the independent failures model, 
such that both M and M.RECOVER are lock-free. 
\end{theorem}

\begin{proof} 
Assume towards a contradiction that such a lock-free implementation exists. 
Assume that process $p$ invokes $M$ with value $x$ and process $q$ 
invokes $M$ with value $y$, for $x,y$ and $z$ as in Definition~\ref{def_distinguishable}.

To prove the theorem, 
we construct an execution in which each process performs an infinite number of steps and $q$ neither crashes nor completes its operation.

Configuration $C_0$, reached after execution $\alpha_{base}$, is bivalent because a solo execution of either $p$ or $q$ from $C_0$ returns $z$. Following a standard valency argument and since we assume that $M$ is lock-free, there is a crash-free execution starting from $C_0$ that leads to a bivalent configuration $C_1$, in which both $p$ and $q$ are about to execute a critical step. It must be that one step leads to a $p$-univalent configuration while the other leads to a $q$-univalent configuration. 

\begin{restatable}{claim}{critical}
\label{criticalsteps}
The critical steps of $p$ and $q$ apply (possibly the same) primitives $f_i$ and $f_j$, 
respectively, to the same base object.
\end{restatable}

\begin{proof}
Consider all possible steps: read, write, crash and $f_1 \dots f_K$. Assume $s_p$ and $s_q$ are critical steps by process $p$ and $q$ respectively, such that $C_1 \circ s_p$ is $p$-univalent while $C_1 \circ s_q$ is $q$-univalent.
\begin{itemize}
    \item Steps $s_p$ and $s_q$ access distinct registers. In this case, these configurations are indistinguishable to $p$ and $q$, that is, $C \circ s_p \circ s_q \stackrel{p,q}{\sim} C \circ s_q \circ s_p$
    \item Step $s_q$ is a crash step then $C \circ s_p  \circ s_q \stackrel{p}{\sim} C \circ s_q  \circ s_p$
    \item Steps $s_p$ and $s_q$ read the same register. Also in this case $C \circ s_p \circ s_q \stackrel{p,q}{\sim} C \circ s_q \circ s_p$
    \item Step $s_p$ writes to some register $r$ step and $s_q$ reads $r$. In this case, $C \circ s_p \stackrel{p}{\sim} C \circ s_q \circ s_p$ holds.
    \item Step $s_p$ applies $f_i$ $1\leq i \leq K$ and step $s_q$ reads $r$. In this case, $C \circ s_p \stackrel{p}{\sim} C \circ s_q \circ s_p$ holds.
    \item Steps $s_p$ and $s_q$ write to the same register. In this case, $C \circ s_p \stackrel{p}{\sim} C \circ s_q \circ s_p$ holds.
    \item Step $s_p$ applies $f_i$, $1\leq i \leq K$, step $s_q$ writes to the same register. In this case, $C \circ s_q \stackrel{q}{\sim} C \circ s_p \circ s_q$ holds.
    \item Step $s_p$ applies $f_i$, $1\leq i \leq K$, step $s_q$ applies $f_j$, $1\leq j \leq K$ each to a different base object $O$, In this case $C \circ s_p \circ s_q \stackrel{p,q}{\sim} C \circ s_q \circ s_p$ holds.
    
\end{itemize}
In each of the above cases, the configurations are indistinguishable to at least one process, and therefore, must have the same valencies. Therefore, it must be that $p$ and $q$ apply $f_i$ and $f_j$ respectively to the same base object.
\end{proof}

Assume, without loss of generality, that $C_1\circ p$ is $p$-univalent while $C_1\circ q$ is $q$-univalent.
We consider two cases: 

\textbf{Case 1: $f_i$ and $f_j$ commute:}
Consider executions $C_2 , C_3$ where $C_2 = C_1\circ p \circ q \circ CRASH_p$ and $C_3 = C_1\circ q \circ p \circ CRASH_p$. Configurations $C_2$ and $C_3$ are reached after $p$ and $q$ each take a step (in different orders) in which they apply their operations to the same base object $O$ and then $p$ crashes.

A solo execution of M.RECOVER by $p$ from both $C_2$ and $C_3$ must complete since $I$ is lock-free. Furthermore, $C_2 \stackrel{p}{\sim} C_3$ holds, because $p$'s response from the primitive $f_i$ is lost, while the value of $O$ is the same in both configurations since $f_i$ and $f_j$ commute. Consequently, an execution of M.RECOVER by $p$ from both $C_2$ and $C_3$ must return the same value. Let $v$ denote this value.

Assume first that $v=z$ and thus $C_3$ is $p$-valent. Configuration $C_1\circ q \circ p$ is $q$-univalent, while $C_3 = C_1\circ q \circ p \circ CRASH_p$ is $p$-valent. However,  $C_1\circ q \circ p \stackrel{q}{\sim} C_3$ holds because $q$ is unaware of $p$’s crash. Consequently, a solo execution of $q$ from $C_3$ must return $z$, that is, $C_3$ is also $q$-valent. This proves that $C_3$ is bivalent.

Assume then that $v\ne z$. We now show that, in this case, $C_2$ is bivalent. Indeed, from this assumption, $C_2$ is $q$-valent, because a solo execution of $q$ after $p$ completes (and returns $v \ne z$) must return $z$ since, from Definition~\ref{def_distinguishable}, exactly one of these two operations must return $z$. However, configuration $C_1 \circ p \circ q$ is $p$-univalent, while $C_2 = C_1 \circ p \circ q \circ CRASH_p \stackrel{q}{\sim} C_1 \circ p \circ q$, therefore a solo execution of $q$ from $C_2$ must return $x$ s.t. $x\ne z$. Thus, $C_2$ is bivalent.


\textbf{Case 2: $f_j$ overwrites $f_i$:}
Consider executions $C_2 , C_3$ where $C_2 = C_1\circ p \circ q \circ CRASH_p$ and $C_3 = C_1 \circ q \circ CRASH_p$.
A solo execution of M.RECOVER by $p$ from both $C_2$ and $C_3$ must complete since \emph{I} is lock-free. Furthermore, $C_2 \stackrel{p}{\sim} C_3$ because $p$'s response from the primitive $f_i$ is lost, while the value of the base object $f_i$ and $f_j$ are applied to is the same in both configurations since $f_j$ overwrites $f_i$. Therefore, an execution of M.RECOVER by $p$ from both $C_2$ and $C_3$ returns the same value. Let $v$ denote this value.

Assume $v = z$ and thus $C_3$ is $p$-valent. 
$C_1 \circ q$ is $q$-univalent, while $C_3 = C_1 \circ q \circ CRASH_p$ is $p$-valent. 
$C_1 \circ q \stackrel{q}{\sim} C_3$ holds because $q$ is unaware of $p$’s crash.
Therefore, a solo execution of $q$ from $C_3$ returns $z$, 
that is, $C_3$ is also $q$-valent. 
This proves that $C_3$ is bivalent.

Assume then that $v\ne z$. We show that in this case $C_2$ is bivalent. Indeed, from our assumption, $C_2$ is $q$-valent, as a solo execution of $q$ after $p$ completes must return $z$ since, from Definition~\ref{def_distinguishable}, exactly one of these two operations must return $z$. However configuration $C_1 \circ p \circ q$ is $p$-univalent and $C_2 = C_1 \circ p \circ q \circ CRASH_p \stackrel{q}{\sim} C_1 \circ p \circ q$, therefore a solo execution of $q$ from $C_2$ must return $x\ne z$. This establishes that $C_2$ is bivalent.

In both cases, this shows that we can keep extending the execution 
obtaining an infinite execution in which neither $p$ nor $q$ 
complete their operations and $q$ performs an infinite number of steps without crashing, 
contradicting the lock-freedom assumption.
\end{proof}



\section{Discussion}

We present two NRL self-implementations of the swap object, one for the system-wide failures model and 
the other for the independent failures model. In both, SWAP operations are wait-free and the recovery code is blocking. In the system-wide failures model, this is a result of delegating the recovery to a single process, while in the independent failures model, it is due to coordination between the recovering process and the other processes. 
We also prove the impossibility of a lock-free implementation of distinguishable operations using read-write and a set of interfering functions, in the independent failures model. 
In particular, this shows that with independent failures, a self-implementation of swap cannot be lock-free.

Our algorithms use $O(m*n)$ space, where $m$ is the number of SWAP invocations in the execution. Bounding memory consumption to $O(n)$ is relatively easy if a recoverable swap operation by one process can wait for operations by other processes to either make progress or fail. An interesting open question is to figure out whether the space complexity of detectable swap self-implementations with wait-free operations can be reduced to $o(m)$ or if $\Omega(m)$ is inherently required. We leave this question for future work.


Finally, it is also important to explore how self-implementations, 
in particular of swap, can be used to turn non-recoverable higher-level objects
into NRL implementations of the same objects.





{
	\bibliographystyle{plainurl}
	\bibliography{main}

\begin{thebibliography}{10}

\bibitem{aguilera2003strict}
Marcos~K Aguilera and Svend Fr{\o}lund.
\newblock Strict linearizability and the power of aborting.
\newblock {\em Technical Report HPL-2003-241}, 2003.

\bibitem{attiya2018nesting}
Hagit Attiya, Ohad Ben-Baruch, and Danny Hendler.
\newblock Nesting-safe recoverable linearizability: Modular constructions for
  non-volatile memory.
\newblock In {\em Proceedings of the 2018 ACM Symposium on Principles of
  Distributed Computing}, pages 7--16, 2018.

\bibitem{ben2020upper}
Ohad Ben-Baruch, Danny Hendler, and Matan Rusanovsky.
\newblock Upper and lower bounds on the space complexity of detectable objects.
\newblock In {\em Proceedings of the 39th Symposium on Principles of
  Distributed Computing}, pages 11--20, 2020.

\bibitem{ben2019delay}
Naama Ben-David, Guy~E Blelloch, Michal Friedman, and Yuanhao Wei.
\newblock Delay-free concurrency on faulty persistent memory.
\newblock In {\em The 31st ACM Symposium on Parallelism in Algorithms and
  Architectures}, pages 253--264, 2019.

\bibitem{ben2022survey}
Naama Ben-David, Michal Friedman, and Yuanhao Wei.
\newblock Survey of persistent memory correctness conditions.
\newblock {\em arXiv preprint arXiv:2208.11114}, 2022.

\bibitem{berryhill2016robust}
Ryan Berryhill, Wojciech Golab, and Mahesh Tripunitara.
\newblock Robust shared objects for non-volatile main memory.
\newblock In {\em 19th International Conference on Principles of Distributed
  Systems (OPODIS 2015)}. Schloss Dagstuhl-Leibniz-Zentrum fuer Informatik,
  2016.

\bibitem{chakrabarti2014atlas}
Dhruva~R Chakrabarti, Hans-J Boehm, and Kumud Bhandari.
\newblock Atlas: Leveraging locks for non-volatile memory consistency.
\newblock {\em ACM SIGPLAN Notices}, 49(10):433--452, 2014.

\bibitem{chan2020recoverable}
David Yu~Cheng Chan and Philipp Woelfel.
\newblock Recoverable mutual exclusion with constant amortized rmr complexity
  from standard primitives.
\newblock In {\em Proceedings of the 39th Symposium on Principles of
  Distributed Computing}, pages 181--190, 2020.

\bibitem{cho2022practical}
Kyeongmin Cho, Seungmin Jeon, and Jeehoon Kang.
\newblock Practical detectability for persistent lock-free data structures.
\newblock {\em arXiv preprint arXiv:2203.07621}, 2022.

\bibitem{cohen2018inherent}
Nachshon Cohen, Rachid Guerraoui, and Igor Zablotchi.
\newblock The inherent cost of remembering consistently.
\newblock In {\em Proceedings of the 30th on Symposium on Parallelism in
  Algorithms and Architectures}, pages 259--269, 2018.

\bibitem{correia2018romulus}
Andreia Correia, Pascal Felber, and Pedro Ramalhete.
\newblock Romulus: Efficient algorithms for persistent transactional memory.
\newblock In {\em Proceedings of the 30th on Symposium on Parallelism in
  Algorithms and Architectures}, pages 271--282, 2018.

\bibitem{correia2020persistent}
Andreia Correia, Pascal Felber, and Pedro Ramalhete.
\newblock Persistent memory and the rise of universal constructions.
\newblock In {\em Proceedings of the Fifteenth European Conference on Computer
  Systems}, pages 1--15, 2020.

\bibitem{david2018log}
Tudor David, Aleksandar Dragojevic, Rachid Guerraoui, and Igor Zablotchi.
\newblock Log-free concurrent data structures.
\newblock In {\em 2018 {USENIX} Annual Technical Conference}, pages 373--386,
  2018.

\bibitem{FatourouKK2022performance}
Panagiota Fatourou, Nikolaos~D Kallimanis, and Eleftherios Kosmas.
\newblock The performance power of software combining in persistence.
\newblock In {\em Proceedings of the 27th ACM SIGPLAN Symposium on Principles
  and Practice of Parallel Programming}, pages 337--352, 2022.

\bibitem{fischer1985impossibility}
Michael~J Fischer, Nancy~A Lynch, and Michael~S Paterson.
\newblock Impossibility of distributed consensus with one faulty process.
\newblock {\em Journal of the ACM (JACM)}, 32(2):374--382, 1985.

\bibitem{friedman2020nvtraverse}
Michal Friedman, Naama Ben-David, Yuanhao Wei, Guy~E Blelloch, and Erez
  Petrank.
\newblock Nvtraverse: In nvram data structures, the destination is more
  important than the journey.
\newblock In {\em Proceedings of the 41st ACM SIGPLAN Conference on Programming
  Language Design and Implementation}, pages 377--392, 2020.

\bibitem{friedman2018persistent}
Michal Friedman, Maurice Herlihy, Virendra Marathe, and Erez Petrank.
\newblock A persistent lock-free queue for non-volatile memory.
\newblock {\em ACM SIGPLAN Notices}, 53(1):28--40, 2018.

\bibitem{golab2018recoverable}
Wojciech Golab.
\newblock Recoverable consensus in shared memory.
\newblock {\em arXiv preprint arXiv:1804.10597}, 2018.

\bibitem{golab2020recoverable}
Wojciech Golab.
\newblock The recoverable consensus hierarchy.
\newblock In {\em Proceedings of the 32nd ACM Symposium on Parallelism in
  Algorithms and Architectures}, pages 281--291, 2020.

\bibitem{golab2017recoverable}
Wojciech Golab and Danny Hendler.
\newblock Recoverable mutual exclusion in sub-logarithmic time.
\newblock In {\em Proceedings of the ACM Symposium on Principles of Distributed
  Computing}, pages 211--220, 2017.

\bibitem{golab2016recoverable}
Wojciech Golab and Aditya Ramaraju.
\newblock Recoverable mutual exclusion.
\newblock In {\em Proceedings of the 2016 ACM Symposium on Principles of
  Distributed Computing}, pages 65--74, 2016.

\bibitem{guerraoui2004robust}
Rachid Guerraoui and Ron~R Levy.
\newblock Robust emulations of shared memory in a crash-recovery model.
\newblock In {\em 24th International Conference on Distributed Computing
  Systems, 2004. Proceedings.}, pages 400--407. IEEE, 2004.

\bibitem{herlihy1991wait}
Maurice Herlihy.
\newblock Wait-free synchronization.
\newblock {\em ACM Transactions on Programming Languages and Systems (TOPLAS)},
  13(1):124--149, 1991.

\bibitem{herlihy1990linearizability}
Maurice~P Herlihy and Jeannette~M Wing.
\newblock Linearizability: A correctness condition for concurrent objects.
\newblock {\em ACM Transactions on Programming Languages and Systems (TOPLAS)},
  12(3):463--492, 1990.

\bibitem{hoseinzadeh2021corundum}
Morteza Hoseinzadeh and Steven Swanson.
\newblock Corundum: Statically-enforced persistent memory safety.
\newblock In {\em Proceedings of the 26th ACM International Conference on
  Architectural Support for Programming Languages and Operating Systems}, pages
  429--442, 2021.

\bibitem{izraelevitz2016failure}
Joseph Izraelevitz, Terence Kelly, and Aasheesh Kolli.
\newblock Failure-atomic persistent memory updates via justdo logging.
\newblock {\em ACM SIGARCH Computer Architecture News}, 44(2):427--442, 2016.

\bibitem{izraelevitz2016linearizability}
Joseph Izraelevitz, Hammurabi Mendes, and Michael~L Scott.
\newblock Linearizability of persistent memory objects under a
  full-system-crash failure model.
\newblock In {\em Distributed Computing: 30th International Symposium, DISC
  2016, Paris, France, September 27-29, 2016. Proceedings 30}, pages 313--327.
  Springer, 2016.

\bibitem{jayanti2019optimal}
Prasad Jayanti, Siddhartha Jayanti, and Anup Joshi.
\newblock Optimal recoverable mutual exclusion using only fasas.
\newblock In {\em Networked Systems: 6th International Conference, NETYS 2018,
  Essaouira, Morocco, May 9--11, 2018, Revised Selected Papers}, pages
  191--206. Springer, 2019.

\bibitem{jayanti2019recoverable}
Prasad Jayanti, Siddhartha Jayanti, and Anup Joshi.
\newblock A recoverable mutex algorithm with sub-logarithmic {RMR} on both cc
  and dsm.
\newblock In {\em Proceedings of the 2019 ACM Symposium on Principles of
  Distributed Computing}, pages 177--186, 2019.

\bibitem{jayanti2023constant}
Prasad Jayanti, Siddhartha Jayanti, and Anup Joshi.
\newblock Constant rmr recoverable mutex under system-wide crashes.
\newblock {\em arXiv preprint arXiv:2302.00748}, 2023.

\bibitem{jayanti2017recoverable}
Prasad Jayanti and Anup Joshi.
\newblock Recoverable {FCFS} mutual exclusion with wait-free recovery.
\newblock In {\em 31st International Symposium on Distributed Computing (DISC
  2017)}. Schloss Dagstuhl-Leibniz-Zentrum fuer Informatik, 2017.

\bibitem{jayanti2022recoverable}
Prasad Jayanti and Anup Joshi.
\newblock Recoverable mutual exclusion with abortability.
\newblock {\em Computing}, 104(10):2225--2252, 2022.

\bibitem{katzan2020recoverable}
Daniel Katzan and Adam Morrison.
\newblock Recoverable, abortable, and adaptive mutual exclusion with
  sublogarithmic rmr complexity.
\newblock {\em arXiv preprint arXiv:2011.07622}, 2020.

\bibitem{KatzanM2021recoverable}
Daniel Katzan and Adam Morrison.
\newblock Recoverable, abortable, and adaptive mutual exclusion with
  sublogarithmic rmr complexity.
\newblock In {\em 24th International Conference on Principles of Distributed
  Systems}, 2021.

\bibitem{li2021detectable}
Nan Li and Wojciech Golab.
\newblock Detectable sequential specifications for recoverable shared objects.
\newblock In {\em 35th International Symposium on Distributed Computing (DISC
  2021)}. Schloss Dagstuhl-Leibniz-Zentrum f{\"u}r Informatik, 2021.

\bibitem{mellor1991algorithms}
John~M Mellor-Crummey and Michael~L Scott.
\newblock Algorithms for scalable synchronization on shared-memory
  multiprocessors.
\newblock {\em ACM Transactions on Computer Systems (TOCS)}, 9(1):21--65, 1991.

\bibitem{nahum2022recoverable}
Liad Nahum, Hagit Attiya, Ohad Ben-Baruch, and Danny Hendler.
\newblock Recoverable and detectable fetch\&add.
\newblock In {\em 25th International Conference on Principles of Distributed
  Systems (OPODIS 2021)}. Schloss Dagstuhl-Leibniz-Zentrum f{\"u}r Informatik,
  2022.

\bibitem{nawab2017dali}
Faisal Nawab, Joseph Izraelevitz, Terence Kelly, Charles~B Morrey~III, Dhruva~R
  Chakrabarti, and Michael~L Scott.
\newblock Dal{\'\i}: A periodically persistent hash map.
\newblock In {\em 31st International Symposium on Distributed Computing (DISC
  2017)}. Schloss Dagstuhl-Leibniz-Zentrum fuer Informatik, 2017.

\bibitem{ramalhete2019onefile}
Pedro Ramalhete, Andreia Correia, Pascal Felber, and Nachshon Cohen.
\newblock Onefile: A wait-free persistent transactional memory.
\newblock In {\em 2019 49th Annual IEEE/IFIP International Conference on
  Dependable Systems and Networks (DSN)}, pages 151--163. IEEE, 2019.

\bibitem{rusanovsky2021flat}
Matan Rusanovsky, Hagit Attiya, Ohad Ben-Baruch, Tom Gerby, Danny Hendler, and
  Pedro Ramalhete.
\newblock Flat-combining-based persistent data structures for non-volatile
  memory.
\newblock In {\em Stabilization, Safety, and Security of Distributed Systems:
  23rd International Symposium, SSS 2021, Virtual Event, November 17--20, 2021,
  Proceedings 23}, pages 505--509. Springer, 2021.

\bibitem{schwalb2015nvc}
David Schwalb, Markus Dreseler, Matthias Uflacker, and Hasso Plattner.
\newblock {NVC}-hashmap: A persistent and concurrent hashmap for non-volatile
  memories.
\newblock In {\em Proceedings of the 3rd VLDB Workshop on In-Memory Data
  Mangement and Analytics}, pages 1--8, 2015.

\bibitem{yang1995fast}
Jae-Heon Yang and Jams~H Anderson.
\newblock A fast, scalable mutual exclusion algorithm.
\newblock {\em Distributed Computing}, 9:51--60, 1995.

\bibitem{zuriel2019efficient}
Yoav Zuriel, Michal Friedman, Gali Sheffi, Nachshon Cohen, and Erez Petrank.
\newblock Efficient lock-free durable sets.
\newblock {\em Proceedings of the ACM on Programming Languages},
  3(OOPSLA):1--26, 2019.

\end{thebibliography}
}

\end{document}